\documentclass{revtex4}
\usepackage{amsmath,amsthm}
\usepackage{setspace}
\usepackage{graphicx}
\newcommand{\bra}[1]{\langle#1|}
\newcommand{\ket}[1]{|#1\rangle}
\usepackage{amsfonts}
\newtheorem{theorem}{Theorem}[section]
\newtheorem{lemma}[theorem]{Lemma}
\newtheorem{proposition}[theorem]{Proposition}
\newtheorem{cor}[theorem]{Corollary}

\theoremstyle{remark}
\newtheorem{remark}[theorem]{Remark}

\theoremstyle{definition}
\newtheorem{definition}[theorem]{Definition}

\theoremstyle{example}
\newtheorem{example}[theorem]{Example}

\theoremstyle{notation}
\newtheorem{notation}[theorem]{Notation}

\begin{document}
\title{Partial order and a $T_0$-topology in a set of finite quantum systems}            
\author{A. Vourdas}
\affiliation{Department of Computing,\\
University of Bradford, \\
Bradford BD7 1DP, United Kingdom}

\begin{abstract}
A `whole-part' theory is developed for a set of finite quantum systems $\Sigma (n)$
with variables in ${\mathbb Z}(n)$.
The partial order `subsystem' is defined, by embedding various attributes 
of the system $\Sigma (m)$ (quantum states, density matrices, etc) 
into their counterparts in the supersystem $\Sigma (n)$ (for $m|n$).
The compatibility of these embeddings is studied.
The concept of ubiquity is introduced for quantities which fit with this structure. 
It is shown that various entropic quantities  are ubiquitous.
The sets of various quantities become $T_0$-topological spaces with the divisor topology, 
which encapsulates fundamental physical properties.
These sets can be converted into directed-complete partial orders (dcpo),
by adding `top elements'. The continuity of various maps among these sets is studied. 
\end{abstract}

\maketitle

\section{Introduction}
In mathematics, after we define a structure, we study its substructures (e.g, subgroups in group theory, etc). 
In this paper we do something similar for finite quantum systems.

There has been much work in the past few years on various aspects of a finite quantum system $\Sigma (n)$
with variables in ${\mathbb Z}(n)$ (the integers modulo $n$).
Reviews of this work have been presented in \cite{FI1,FI2,FI3,FI4,FI5,FI6}.
In addition to that there has also been much work on multipartite systems (and in this paper 
we are interested in the case where the component systems are finite dimensional), and in particular 
on their classical and quantum correlations (reviewed in \cite{Horo}).

In this paper we discuss a `whole-part' theory in the context of finite quantum systems 
(the term `mereology' is also used in a philosophical context for whole-part type of ideas).
We first introduce a partial order in the set $\{\Sigma (n)\}$ based on the concept `subsystem' 
(and supersystem).
For $m|n$, the basic attributes associated with $\Sigma (m)$ 
(e.g., quantum states, density matrices, observables, etc) are embedded into the corresponding 
attributes in $\Sigma (n)$. An important requirement here is the compatibility between the various 
embeddings, so that we get a self-consistent structure. 
There are many attributes characterizing a quantum system, and we require compatibility between the 
corresponding partial orders.   

Apart from the basic attributes, there are many other quantities used in the description of $\Sigma (m)$ 
(e.g., entropic quantities, Wigner functions, etc) which we characterize as ubiquitous and nonubiquitous.
For an ubiquitous quantity, the calculation of its value for a state in $\Sigma (m)$, 
gives the same answer as the calculation of its value 
for the counterpart (`same') state within any of the supersystems $\Sigma (n)$ (where $m|n$). 
A nonubiquitous quantity has local use within $\Sigma (m)$, and its calculation (for the counterpart state)
within any of the supersystems $\Sigma (n)$, gives a different result.
Ubiquitous quantities fit with the poset structure of regarding smaller systems as subsystems of larger ones,
while  nonubiquitous quantities do not fit to that scheme. 
We prove that various entropic quantities and also various quantities in phase space
(e.g., Wigner and Weyl functions) are ubiquitous.

After we make precise these concepts,  
it is natural to explore if there is continuity of a quantity (e.g., entropy) 
as a function of the dimension of the system $n$.
A prerequisite for any discussion of continuity, is to
introduce topologies in the sets.
There are many topologies that can be defined on a given set, and for physical reasons we use the 
divisor topology (e.g.,\cite{counterexamples,top1,top2}), where
an open (resp. closed) set contains the quantity in some systems and all their subsystems 
(resp. all their supersystems). This is a $T_0$-topology.

There is a special family of partial orders which is very useful. They are the
directed-complete partial orders (dcpo), and they play an important role  in theoretical 
computer science \cite{D1,D2,D3}. The set $\{\Sigma (n)\}$ is not a dcpo, but by
adding suitable `top elements' we can convert it to a dcpo.
This links finite quantum systems with theoretical computer science and quantum computing.

In section 2, we discuss the divisor topology. In section 3, we discuss Heisenberg-Weyl groups and 
symplectic transformations (in a group theoretical context).
In section 4, we derive some results on matrices which are used in the proof of some propositions later.
In section 5, we discuss briefly the quantum formalism on a finite quantum system $\Sigma (n)$
with variables in ${\mathbb Z}(n)$. In section 6, we define the concepts of subsystem and supersystem
by embedding several attributes of $\Sigma (m)$ into their counterparts in $\Sigma (n)$.
We also study
the compatibility of these formalisms. Then the set of all systems $\Sigma (n)$ becomes a poset 
with the partial order `subsystem'.
In section 7, we introduce the concept of ubiquity
which identifies quantities that fit with the poset structure.
In section 8, we make some of these sets topological spaces with the divisor topology, and discuss 
the physical meaning of the topology.
In section 9 we extend these ideas to bipartite systems.
In section 10, we add `top elements' to the set $\{\Sigma (n)\}$ so that it becomes a dcpo.
We conclude in section 11, with a discussion of our results. 
Throughout the paper, we discuss in detail the physical meaning of the mathematical formalism.

\section{Divisor topology}

\begin{notation}
\mbox{}
\begin{itemize}
\item[(1)]
${\mathbb R}$ denotes the set of real numbers;
${\mathbb R}_0^+$ the non-negative real numbers;
${\mathbb Z}$ the integers; ${\mathbb Z}^+$ the positive integers; ${\mathbb Z}_0^+$ the non-negative integers;
and $\Pi$ the prime numbers.
\item[(2)]
${\rm GCD}(r,s)$ and ${\rm LCM}(r,s)$ are the greatest common divisor and least common multiplier correspondingly, of the integers $r,s$.
\item[(3)]
${\mathbb Z}(n)$ is the ring of integers modulo $n$. Also
\begin{equation}
\omega _n(\alpha)=\exp \left (i\frac {2\pi \alpha}{n}\right );\;\;\;\;\alpha \in {\mathbb Z}(n)
\end{equation}
${\mathbb Z}^*(n)$ is the reduced system of residues modulo $n$, and contains the units (invertible elements) 
of ${\mathbb Z}(n)$.
Its cardinality is given by the Euler function $\varphi(n)$.
\item[(4)]
If $F,G$ are groups, $F\le G$ denotes that $F$ is a subgroup of $G$ (in this case we call $G$ a supergroup of $F$).
\item[(5)]
$r|s$ denotes that $r$ is a divisor of $s$.
We will also use the notation 
$(r_1,r_2)|(s_1,s_2)$ to indicate that $r_1|s_1$ and also $r_2|s_2$.
\end{itemize}
\end{notation}
\begin{definition}
\mbox{}
\begin{itemize}
\item[(1)]
A poset is a set ${\mathbb A}$ with a binary relation $\prec$ such that
\begin{itemize}
\item
$a\prec a$, for all $a\in {\mathbb A}$ (reflexivity)
\item
if $a\prec b$ and $b\prec a$, then $a=b$ (antisymmetry)
\item
if $a\prec b$ and $b\prec c$, then $a\prec c$ (transitivity)
\end{itemize}
\item[(2)]
An element $m\in {\mathbb A}$ is called minimal, if there is no element $a\in {\mathbb A}$ such that $a\prec m$.
In a dual way we define the maximal elements.
\item[(3)]
A lower bound of a subset ${\mathbb B}$ of the poset ${\mathbb A}$, is an element $a\in {\mathbb A}$ such that $a\prec b$ for all $b\in {\mathbb B}$. 
If the set of all lower bounds of ${\mathbb B}$ has a largest element, it is called the infimum of ${\mathbb B}$.
In a dual way we define the upper bounds and the supremum of ${\mathbb B}$.
\item[(4)]
A directed poset is a poset such that
for $a,b\in {\mathbb A}$ there exists $c\in {\mathbb A}$ such that $a\prec b$ and $b\prec c$.
\item[(5)]
A function $f:\;\;{\mathbb A}\;\rightarrow \;{\mathbb B}$, where ${\mathbb A}, {\mathbb B}$ are posets,
is a monotone (order preserving), if $a_1\prec a_2$ implies that $f(a_1)\prec f(a_2)$.
\end{itemize}
\end{definition}
Throughout the paper we have various posets and for simplicity we use the same symbol $\prec$ for all orders
(but for subgroups we use the symbol $\le$).

\subsection{The directed poset ${\mathbb X}$ as a topological space  with the divisor topology}
Let ${\mathbb X}$ be the set
\begin{eqnarray}\label{ty}
{\mathbb X}=\{2,3,...\}
\end{eqnarray}
${\mathbb X}$ is a directed poset with division as partial order (i.e. $m\prec n$ if $m|n$).
Also ${\mathbb X}\times {\mathbb X}$ is a directed poset with division as partial order 
(i.e. $(m_1,m_2)\prec (n_1,n_2)$ if $(m_1,m_2)|(n_1,n_2)$).
\begin{remark}\label{567}
The number $1$ could be included in ${\mathbb X}$, but then the trivial quantum system $\Sigma (1)$ with variables in ${\mathbb Z}(1)=\{0\}$
and one-dimensional Hilbert space $H(1)$, would have to be included in the set of quantum systems $\{\Sigma (n)\}$.
The physical importance of such system with one-dimensional Hilbert space is limited, and
we have chosen to exclude it from the formalism.
\end{remark}
\begin{proposition}\label{q2}
\mbox{}
\begin{itemize}
\item[(1)]
The set of minimal elements in ${\mathbb X}$ is the set of prime numbers.
\item[(2)]
The supremum of any finite subset ${\mathbb A}$ of ${\mathbb X}$ is the least common multiplier of all the elements of  ${\mathbb A}$.
\end{itemize}
Analogous statements can be made for ${\mathbb X}\times {\mathbb X}$.
\end{proposition}
\begin{proof}
\mbox{}
\begin{itemize}
\item[(1)]
If $p$ is a prime number, there is no element in $a\in {\mathbb X}$ such that $a|p$
\item[(2)]
If $u$ is the least common multiplier of all the elements of  ${\mathbb A}$, then all the multiples $Nu$ are upper bounds of $A$ and $u$ is the lowest of them.
\end{itemize}
\end{proof}

\begin{definition}
\mbox{}
\begin{itemize}
\item[(1)]
The topological space $({\mathbb X}, {\mathfrak T}_{\mathbb X})$
is the set ${\mathbb X}$ with the divisor topology ${\mathfrak T}_{\mathbb X}$ generated by the base
\begin{eqnarray}
B_{\mathbb X}=\{\emptyset, {\mathbb X}, {U}_{\mathbb X}(n)\;|\;n=2,3,...\}
\end{eqnarray}
where 
\begin{eqnarray}
U_{\mathbb X}(n)=\{m\in {{\mathbb X}}\;|\;m|n\};\;\;\;\;n=2,3,...
\end{eqnarray}
All unions of elements of $B_{\mathbb X}$ are open sets and they are 
elements of ${\mathfrak T}_{\mathbb X}$.
The closed sets are the complements of the open sets in ${\mathbb X}$, and they are all the 
${\mathbb X}-U_{\mathbb X}(n)$ together with all their intersections.
\item[(2)]
The topological space $({\mathbb X}\times {\mathbb X}, {\mathfrak T}_{{\mathbb X}\times {\mathbb X}})$
is the set ${\mathbb X}\times {\mathbb X}$ with the product (Tychonoff) topology
${\mathfrak T}_{{\mathbb X}\times {\mathbb X}}$ generated by the base 
\begin{eqnarray}
B_{{\mathbb X}\times {\mathbb X}}=\{\emptyset, \;{\mathbb X}\times {\mathbb X}, \;
U_{{\mathbb X}\times {\mathbb X}}(n_1,n_2)\;|\;n_1,n_2=2,3,...\}
\end{eqnarray}
where
\begin{eqnarray}
U_{{\mathbb X}\times {\mathbb X}}(n_1,n_2)=\{(m_1,m_2)\in {{\mathbb X}\times {\mathbb X}}\;|\;
(m_1,m_2)|(n_1,n_2)\}={U}_{\mathbb X}(n_1)\times {U}_{\mathbb X}(n_2).
\end{eqnarray}
\end{itemize}
\end{definition}
\begin{remark}
We note that
\begin{eqnarray}
&&{\rm GCD}(n,m)=1\;\;\rightarrow\;\;U_{\mathbb X}(n) \cap U_{\mathbb X}(m)=\emptyset\nonumber\\
&&{\rm GCD}(n,m)>1\;\;\rightarrow\;\;U_{\mathbb X}(n) \cap U_{\mathbb X}(m)=U_{\mathbb X}({\rm GCD}(n,m))
\end{eqnarray}

\end{remark}
\begin{remark}
In a topological space the intersection of a finite number of open sets is an open set.
Here the restriction to a finite number of open sets is not needed because each point $n$ has a 
smallest neibourhood which is ${U}_{\mathbb X}(n)$. 
Therefore open and closed sets satisfy exactly the same conditions (Alexandrov topology).
\end{remark}
\begin{remark}\label{rema}
A poset ${\mathbb A}=\{a_2,a_3,...\}$ is order isomorpic to ${\mathbb X}$ if the map
\begin{eqnarray}
f:\;{\mathbb X}\;\rightarrow\;{\mathbb A};\;\;\;\;f(n)=a_n
\end{eqnarray}
is a bijection and $a_m\prec a_n$ if and only if $m|n$.
In this case both $f$ and $f^{-1}$ are monotone functions.

We make ${\mathbb A}$ a topological space with topology ${\mathfrak T}_{\mathbb A}$
generated by the base
\begin{eqnarray}
B_{\mathbb A}=\{\emptyset, {\mathbb A}, {U}_{\mathbb A}(a_n)\;|\;n=2,3,...\},
\end{eqnarray}
where
\begin{eqnarray}
U_{\mathbb A}(a_n)=\{a_m\in {\mathbb A}\;|\;a_m\prec a_n\};\;\;\;\;n=2,3,...
\end{eqnarray}
The topological space $({\mathbb A}, {\mathfrak T}_{\mathbb A})$ is homeomorphic to the topological space
$({\mathbb X}, {\mathfrak T}_{\mathbb X})$ (we denote this as (${\mathbb A}, 
{\mathfrak T}_{\mathbb A}) \sim ({\mathbb X}, {\mathfrak T}_{\mathbb X})$).

Similar remark can be made for ${\mathbb A}\times {\mathbb A}$. Throughout the paper we give 
several examples of such maps.
\end{remark}
The following properties are known \cite{counterexamples} and we give them without proof.
We will see later that in our context, they reflect fundamental physical aspects of the relationship 
between a finite system and its subsystems and 
supersystems (section \ref{QQ}).

\begin{proposition}\label{propo}
\mbox{}
\begin{itemize}
\item[(1)]
${\mathbb X}$ is a $T_0$-space (Kolmogorov), but it is not a $T_1$-space.
The same is true for ${\mathbb X}\times {\mathbb X}$.
\item[(2)]
The set $\Pi$ of prime numbers is dense in ${\mathbb X}$. 
The set $\Pi \times \Pi$ is dense in ${\mathbb X}\times {\mathbb X}$.
\item[(3)]
The closure $\overline {\{n\}}$ of $\{n\}$ in ${\mathbb X}$, consists of all multiples of $n$.
The closure of $\{(n_1,n_2)\}$ in ${\mathbb X}\times {\mathbb X}$, 
is $\{(K_1n_1,K_2n_2)\;|\;K_1,K_2\in {\mathbb Z}\}$.
\end{itemize}
\end{proposition}
\begin{example}\label{315}
In the topological space $({\mathbb X}, {\mathfrak T}_{\mathbb X})$, the
\begin{eqnarray}
U_{\mathbb X}(6)=\{2,3,6\};\;\;\;\;U_{\mathbb X}(9)=\{3,9\};\;\;\;\;U_{\mathbb X}(6)\cup U_{\mathbb X}(9)=\{2,3,6,9\}
\end{eqnarray}
are examples of open sets. Their complements ${\mathbb X}-U_{\mathbb X}(6)$, ${\mathbb X}-U_{\mathbb X}(9)$, ${\mathbb X}-(U_{\mathbb X}(6)\cup U_{\mathbb X}(9))$, are examples of closed sets. 
Another example of a closed set is the $\overline {\{3\}}=\{3,6,9,...\}$. This is the closure of $\{3\}$ because it is the smallest closed set containing $3$.
\end{example}
\begin{remark}
The divisor function $\sigma _k(n)$ is the sum of the $k$-powers of all divisors of $n$ 
(including $1$ and $n$)\cite{AR}:
\begin{eqnarray}
\sigma _k(n)=\sum _{d|n}d^k
\end{eqnarray} 
The cardinality of $U_{\mathbb X}(n)$ is $\sigma _0(n)-1$ (the $1$ is not included in ${\mathbb X}$).
\end{remark}

\subsection{The topological space  of the additive groups ${\mathbb Z}(n)$}

We consider ${\mathbb Z}(n)$ as additive groups.
Then $m|n$ implies that ${\mathbb Z}(m)\le {\mathbb Z}(n)$. 
The embedding ${\mathfrak K}_{mn}$ of ${\mathbb Z}(m)$ into ${\mathbb Z}(n)$ is given by the injection
\begin{eqnarray}\label{9}
{\mathfrak K}_{mn}:\;\; {\mathbb Z}(m)\ni \alpha\;\rightarrow\; d\alpha \in {\mathbb Z}(n);\;\;\;\;d=\frac{n}{m}
;\;\;\;\;m|n.
\end{eqnarray}
It is compatible in the sense that 
if $m|n|\ell$ then
${\mathfrak K}_{n\ell}\circ {\mathfrak K}_{mn}={\mathfrak K}_{m\ell}$.
\begin{remark}\label{zz}
As a ring ${\mathbb Z}(m)$ is not a subring of ${\mathbb Z}(n)$ ($d\alpha$ times $d\beta$ is 
not $d\alpha \beta$).
This is the cause of difficulties below, in embeddings of structures that use multiplication 
like the Heisenberg-Weyl group, Wigner functions, etc. 
\end{remark}

Let ${\mathfrak Z}$ be the directed poset
\begin{eqnarray}
{\mathfrak Z}=\{{\mathbb Z}(2),{\mathbb Z}(3),...\}
\end{eqnarray}
with $\le $ (subgroup) as partial order.
The map
\begin{eqnarray}
f:\;\;{\mathbb X}\;\;\rightarrow\;\;{\mathfrak Z};\;\;\;\;f(n)={\mathbb Z}(n)
\end{eqnarray}
is a bijection and ${\mathbb Z}(m)\le {\mathbb Z}(n)$ if and only if $m|n$. Consequently,
${\mathfrak Z}$ can be viewed as a topological space
with the divisor topology ${\mathfrak T}_{\mathfrak Z})$, as discussed in remark \ref{rema}.
The topological space $({\mathfrak Z}, {\mathfrak T}_{\mathfrak Z})$ is homeomorphic to the topological space
$({\mathbb X}, {\mathfrak T}_{\mathbb X})$.

An open (resp. closed) set in this topology contains some groups and all their subgroups (resp. supergroups).
As an example, we consider the open set
\begin{eqnarray}
U_{\mathfrak Z}(6)\cup U_{\mathfrak Z}(8)=\{{\mathbb Z}(2),{\mathbb Z}(3),{\mathbb Z}(4),{\mathbb Z}(6),{\mathbb Z}(8)\}
\end{eqnarray}
This contains the groups ${\mathbb Z}(6)$ and ${\mathbb Z} (8)$ and also their subgroups
${\mathbb Z}(2)$, ${\mathbb Z}(3)$ and ${\mathbb Z}(4)$. 
Also the closed set $\overline {\{{\mathbb Z}(5)\}}=\{{\mathbb Z}(5),{\mathbb Z}(10),{\mathbb Z}(15),...\}$
which is the closure of $\{{\mathbb Z}(5)\}$ contains all the supergroups of ${\mathbb Z} (5)$.

\subsection{A class of continuous functions}

\begin{proposition}\label{1000}
Let ${\cal N}:\;{\mathbb X}\;\rightarrow \;{\mathbb Y}\subseteq {\mathbb X}$ be a monotone function,
where ${\mathbb X}$ is the topological space defined earlier and ${\mathbb Y}$ is 
a topological space with the induced topology (i.e.,
its open sets are the intersections of ${\mathbb Y}$ with the open sets of ${\mathbb X}$).
Then ${\cal N}$ is a continous function. 
\end{proposition}
\begin{proof}
For an arbitrary ${\cal N}(u)\in {\mathbb Y}$, we consider a neighbourhood ${\cal V}$ containing ${\cal N}(u)$. 
Then 
\begin{eqnarray}
U_{\mathbb Y}[{\cal N}(u)]=\{{\cal N}(m)\in {\mathbb Y}|{\cal N}(m)|{\cal N}(u)\}
\end{eqnarray}
is a subset of any open set containing ${\cal N}(u)$, and therefore of ${\cal V}$.
We next consider the open set $U_{\mathbb X}(u)$ which contains the divisors $r$ of $u$.
Since ${\cal N}$ is monotone function ${\cal N}(r)|{\cal N}(u)$ and therefore ${\cal N}(r)\in U_{\mathbb Y}[{\cal N}(u)] \subseteq {\cal V}$.
Therefore for an arbitrary neighbourhood ${\cal V}$ containing ${\cal N}(u)$,
we have found an open set $U_{\mathbb X}(u)$ containing $u$, such that ${\cal N}[U_{\mathbb X}(u)] \subseteq {\cal V}$.
This is true for any ${\cal N}(u)$, and therefore the function ${\cal N}$ is continuous.
\end{proof}

\begin{example}\label{1001}
We give some arithmetical functions \cite{AR} which have been used in the study of 
finite quantum systems.
They are functions from ${\mathbb X}$ to a subset of ${\mathbb X}$, and we use proposition \ref{1000} to 
prove that they are continuous.

The Jordan totient function is defined as:
\begin{eqnarray}\label{V1}
J_k(n)=n^k\prod _{p|n}\left (1-\frac{1}{p^k}\right )
\end{eqnarray}
We note that for $k=1$ this is the Euler totient function $J_1(n)=\varphi(n)$.
In the context of finite quantum systems, this function has been used in \cite{VB}, where 
it has been shown that 
the order of the relevant symplectic group is $|Sp(2,{\mathbb Z}(n))|=nJ_2(n)$,
and this has been used for quantum tomography.

Another related function is the Dedekind psi function
\begin{eqnarray}\label{V2}
\psi(n)=\frac{J_2(n)}{\varphi(n)}
\end{eqnarray}
In the context of finite quantum systems, this function has been used in \cite{SV}, 
where it has been shown that the number ${\cal L}(n)$ of `maximal lines through the origin' in the
${\mathbb Z}(n)\times {\mathbb Z}(n)$ phase space is $\psi (n)$ and this has been used in other calculations.
 
Both the Jordan totient function of Eq.(\ref{V1}) and the Dedekind psi function of Eq.(\ref{V2}) 
are monotone functions and therefore they are both continuous functions.

\end{example}
\begin{remark}
A continuous function from ${\mathbb X}$ to a
Hausdorff topological space ${\Phi}$  is fully determined by its values on a dense subset 
(e.g., \cite{Willard, BOU}). 
According to proposition \ref{propo}, the subset $\Pi$ of prime numbers 
is dense in ${\mathbb X}$, but 
in the cases discussed above ${\mathbb Y}$ is a $T_0$-space which is not 
Hausdorff and therefore this statement is not applicable. 
\end{remark}

\section{Heisenberg-Weyl groups and symplectic groups}\label{HW}

\subsection{Pontryagin duality}

In quantum mechanics when the position takes values in an Abelian group $A$, the momenta take values in its
Pontryagin dual group $B$ (which contains the characters $\chi _b(a)$ where $a\in A$).
We consider a subgroup $A_1$ of $A$.
The annihilator of $A_1$ is a subgroup $B_1$ of $B$, such that for all $a\in A_1$ and all $b\in B_1$,
we get $\chi _b (a)=1$.
The theory of Pontryagin duality (e.g., \cite{Po}) proves that
the Pontryagin dual group of $A_1$ is isomorphic to $B/B_1$, and  
the Pontryagin dual group of $A/A_1$ is isomorphic to $B_1$.

When $A={\mathbb Z}(n)$ then its
Pontryagin dual group is $B\cong {\mathbb Z}(n)$.
The characters in this case are $\chi _b(a)=\omega _n(\alpha \beta)$.
We next consider the subgroup $A_1={\mathbb Z}(m)\le A$, where $m|n$.
In this case the annihilator of $A_1$ is the subgroup $B_1={\mathbb Z}(d)$ of $B$,
where $d=n/m$.
Then the Pontryagin dual group of $A_1={\mathbb Z}(m)$ is isomorphic to 
$B/B_1={\mathbb Z}(n)/{\mathbb Z}(d)\cong {\mathbb Z}(m)$ and  
the Pontryagin dual group of $A/A_1={\mathbb Z}(n)/{\mathbb Z}(m)\cong {\mathbb Z}(d)$ 
is isomorphic to $B_1={\mathbb Z}(d)$.

For calculations in the present context, this means the following.
The elements of the subgroup $A_1={\mathbb Z}(m)$ of $A$ can be written as $\alpha =d\alpha'$
where $\alpha '=0,...,m-1$. 
The elements of $B/B_1={\mathbb Z}(n)/{\mathbb Z}(d)$
are the cosets $0({\rm mod}\;d),...,m-1({\rm mod}\;d)$.
Then
\begin{eqnarray}\label{xx}
\omega _n(\alpha \beta)=\omega _m(\alpha'\beta);\;\;\;\;
\alpha =d\alpha';\;\;\;\;\alpha'=0,...,m-1;\;\;\;\;
\beta =0,...,m-1({\rm mod}\;d)
\end{eqnarray}

\subsection{The Heisenberg-Weyl group $HW[{\mathbb Z}(n)]$}

We consider the matrices
\begin{eqnarray}
D_n(\alpha,\beta,\gamma)=\left(\begin{array}{ccc}
1&-\beta &\gamma \\
0&1&\alpha \\
0&0&1\\
\end{array}
\right );\;\;\;\;\alpha,\beta,\gamma \in {\mathbb Z}(n)
\end{eqnarray}
where
\begin{eqnarray}\label{bn}
D_n(\alpha _1,\beta _1,\gamma _1)D_n(\alpha _2,\beta _2,\gamma _2)=
D_n(\alpha _1+\alpha _2,\beta _1+\beta _2,\gamma _1+\gamma _2-\alpha _2\beta _1)
\end{eqnarray}
These matrices form the $HW[{\mathbb Z}(n)]$ group.

We denote as $HW_1[{\mathbb Z}(n)]$ the subgroup of $HW[{\mathbb Z}(n)]$,
which consists of the $D_n(\alpha,0,0)$
(clearly $HW_1[{\mathbb Z}(n)]\cong {\mathbb Z}(n)$).
We also denote as $HW_2[{\mathbb Z}(n)]$ the subgroup which consists of the $D_n(0,\beta,0)$
(and $HW_2[{\mathbb Z}(n)]\cong {\mathbb Z}(n)$)
For applications to quantum mechanics it is essential that $HW_1[{\mathbb Z}(n)]$ is Pontryagin dual to
$HW_2[{\mathbb Z}(n)]$, 
because this allows one of the groups to be related to displacements in positions and the other to 
displacements in momenta.
And indeed $HW_1[{\mathbb Z}(n)]\cong {\mathbb Z}(n)$ is Pontryagin dual to
$HW_2[{\mathbb Z}(n)] \cong {\mathbb Z}(n)$.

For $m|n$, we have the following embedding
from $HW[{\mathbb Z}(m)]$ to $HW[{\mathbb Z}(n)]$:
\begin{eqnarray}
&&{\mathfrak X}_{mn}:\;{D}_m(\alpha, \beta, \gamma)\;\;\rightarrow\;\;{D}_n(d\alpha, \beta, d\gamma);\;\;\;\;
d=\frac{n}{m}
\end{eqnarray}
Here $\alpha$ is multiplied by $d$ and in this sense it 
appears to be treated differently from $\beta$. This is related to
the fact that $\alpha$ and $\beta$ belong to groups which are Pontryagin dual to each other (although, in the 
special case that we consider, they are isomorphic to each other).
In ${D}_n(d\alpha, \beta, d\gamma)$ the
$d\alpha,d\gamma$ take values in the subgroup ${\mathbb Z}(m)$ 
of ${\mathbb Z}(n)$, and $\beta$ takes values in 
its Pontryagin dual group, which as we explained earlier is
${\mathbb Z}(n)/{\mathbb Z}(d)\cong {\mathbb Z}(m)$.

${\mathfrak X}_{mn}$ maps the 
product of two matrices in $HW[{\mathbb Z}(m)]$, into the product of the corresponding matrices in 
$HW[{\mathbb Z}(n)]$.
These maps are compatible in the sense that if $m|n|\ell$ then
${\mathfrak X}_{n\ell}\circ {\mathfrak X}_{mn}={\mathfrak X}_{m\ell}$.

We next consider the set of the Heisenberg-Weyl groups 
$\mathfrak {HW}=\{HW[{\mathbb Z}(2)], HW[{\mathbb Z}(3)],...\}$.
The map 
\begin{eqnarray}\label{aa}
f:\;{\mathbb X}\;\rightarrow \;\mathfrak {HW};\;\;\;\;f(n)=HW[{\mathbb Z}(n)]
\end{eqnarray}
is a bijection and $ HW[{\mathbb Z}(m)]\le  HW[{\mathbb Z}(n)]$ if and only if $m|n$. Consequently,
$\mathfrak {HW}$ can be viewed as a topological space
with the divisor topology ${\mathfrak T}_{\mathfrak {HW}})$, as discussed in remark \ref{rema}.
The topological space $(\mathfrak {HW}, {\mathfrak T}_{\mathfrak {HW}})$ is homeomorphic to the topological space
$({\mathbb X}, {\mathfrak T}_{\mathbb X})$.

\subsection{The symplectic group $Sp(2,{\mathbb Z}(n))$}

The $Sp(2,{\mathbb Z}(n))$ group consists of matrices of the type
\begin{eqnarray}\label{SY}
s_n(\kappa ,\lambda|\mu, \nu)=\left(\begin{array}{cc}
\kappa& \lambda\\
\mu& \nu\\
\end{array}
\right)
;\;\;\;\;\;\;\kappa \nu-\lambda \mu=1\;({\rm mod}\; n);\;\;\;\;\;
\kappa, \lambda, \mu, \nu \in {\mathbb Z}(n).
\end{eqnarray}
For $m|n$, we have an embedding
from $Sp(2,{\mathbb Z}(m))$ to $Sp(2,{\mathbb Z}(n))$, with
\begin{eqnarray}\label{SY}
&&\mathfrak S_{mn}:\;\;s_m(\kappa ,\lambda|\mu, \nu)\;\rightarrow\;s_n(d\kappa ,\lambda|d\mu, \nu)\nonumber\\
&&\kappa \nu-\lambda \mu=1\;({\rm mod}\; m);\;\;\;\;(d\kappa) \nu-\lambda (d\mu)=1\;({\rm mod}\; n)
\end{eqnarray}
Here $\kappa, \mu$ are multiplied by $d$ and in this sense they 
appear to be treated differently from $\lambda,\nu$. 
In section \ref{NN} we will see that in a quantum mechanical context
$\kappa, \mu$ are related to displacements in positions, and $\lambda,\nu$  are related to displacements in 
momenta.
Therefore $\kappa, \mu$
belong to a group which is Pontryagin dual to the group where
the $\lambda,\nu$ belong (but in our case the two groups are isomorphic to each other).
In $s_n(d\kappa ,\lambda|d\mu, \nu)$ the
$d\kappa, d\mu$ take values in the subgroup ${\mathbb Z}(m)$ 
of ${\mathbb Z}(n)$ and $\lambda, \nu$ takes values in the Pontryagin dual group which as we explained earlier is
${\mathbb Z}(n)/{\mathbb Z}(d)\cong {\mathbb Z}(m)$. 

$\mathfrak S_{mn}$ maps the product of two matrices in $Sp(2,{\mathbb Z}(m))$,
into the product of the corresponding matrices in 
$Sp(2,{\mathbb Z}(n))$. Also these maps are compatible, i.e., for $m|n|\ell$ we get
${\mathfrak S}_{n\ell}\circ {\mathfrak S}_{mn}={\mathfrak S}_{m\ell}$.

The set of the symplectic groups
$\{Sp(2,{\mathbb Z}(2)), Sp(2,{\mathbb Z}(3)),...\}$ can become a topological space with the divisor topology,
in a way analogous to that discussed earlier for the Heisenberg-Weyl groups..

\section{Matrices}

Here we summarize some results on matrices which are used in proofs of various propositions later.

For $m|n$, the permutation $\tau_{n,m}$ of the set $\{0,1,...,n-1\}$ is :
\begin{eqnarray}
&&r=0,...,m-1\;\rightarrow\;\tau _{n,m}(r)=rd;\;\;\;\;\;
d=\frac{n}{m};\;\;\;\;\nonumber\\
&&r=m,...,n-1\;\rightarrow\;\tau _{n,m}(r)=r-m+\left [\frac{r-m}{d}\right ]
\end{eqnarray} 
where $[(r-m)/d]$ denotes the smallest integer which is larger than $(r-m)/d$ (i.e., the integral part of $(r-m)/d$ plus one). 

For $m|n$, we consider the $n\times n$ matrix
\begin{eqnarray}
[A_m]_n\equiv \left(\begin{array}{cc}
A_m&0\\
0& 0\\
\end{array}
\right).
\end{eqnarray}
where $A_m$ is a $m\times m$ matrix and the rest of the elements are equal to zero.
The indices of the $[A_m]_n$ matrix take values from $0$ up to $n-1$.
We use the notation ${\mathfrak I}_{mn}(A_m)$ for the $n\times n$ matrix, 
with $[\tau _{n,m}(r),\tau_{n,m}(s)]$ element equal to the $(r,s)$ element of $[A_m]_n$:
\begin{eqnarray}
{\mathfrak I}_{mn}(A_m)[\tau _{n,m}(r),\tau_{n,m}(s)]=[A_m]_n(r,s)
\end{eqnarray}
The matrix ${\mathfrak I}_{mn}(A_m)$ is related to the matrix $[A_m]_n$, through a permutation 
$\tau$ of both columns and rows.
The matrices ${\mathfrak I}_{mn}(A_m)$ and $A_m$ contain the same information and they have the same rank.
\begin{lemma}\label{AA}
For $m|n$,
\begin{itemize}
\item[(1)]
\begin{eqnarray}\label{67}
{\mathfrak I}_{mn}(A_m){\mathfrak I}_{mn}(B_m)={\mathfrak I}_{mn}(A_mB_m)
\end{eqnarray}
\item[(2)]
The $n$ eigenvalues of ${\mathfrak I}_{mn}(A_m)$ are the $m$ eigenvalues of $A_m$ 
plus $n-m$ zeros. Therefore ${\rm Tr}[{\mathfrak I}_{mn}(A_m)]={\rm Tr}(A_m)$.
\end{itemize}
\end{lemma}
\begin{proof}
\mbox{}
\begin{itemize}
\item[(1)]
We have
\begin{eqnarray}
&&{\mathfrak I}_{mn}(A_mB_m)[\tau _{n,m}(r),\tau_{n,m}(s)]=[A_mB_m]_n(r,s)=\sum _t[A_m]_n(r,t)[B_m]_n(t,s)\nonumber\\&&=
\sum _t {\mathfrak I}_{mn}(A_m)[\tau _{n,m}(r),\tau_{n,m}(t)]{\mathfrak I}_{mn}(B_m)[\tau _{n,m}(t),\tau_{n,m}(s)]
\end{eqnarray}
As $t$ takes all values in ${\mathbb Z}(n)$, the $\tau_{n,m}(t)$ also takes all values in ${\mathbb Z}(n)$, and from this follows Eq.(\ref{67}).
\item[(2)]
In general the eigenvalues of a matrix change after a permutation of columns or rows, 
but here we perform the same permutations on both columns and rows. 
The diagonal elements of $[A_m]_n$ remain as diagonal elements in
${\mathfrak I}_{mn}(A_m)$ and consequently the characteristic equation ${\rm det}([A_m]_n-\lambda{\bf 1}_n)$
is the same with the characteristic equation ${\rm det}({\mathfrak I}_{mn}(A_m)-\lambda{\bf 1}_n)$.
\end{itemize}
\end{proof}

\subsection{Bipartite tensors}

Let $(m_1,m_2)|(n_1,n_2)$ and 
\begin{eqnarray}
&&J_1=\{0,...,m_1-1\};\;\;\;\;\;I_1=\{1,...,n_1-1\}\nonumber\\
&&J_2=\{0,...,m_2-1\};\;\;\;\;\;I_2=\{1,...,n_2-1\}
\end{eqnarray}
Also let $A_{m_1,m_2}(r_1,r_2|s_1,s_2)$ be a tensor
with indices $r_1,s_1\in J_1$ and $r_2,s_2\in J_2$. 
We define the tensor
$[A_{m_1,m_2}]_{n_1,n_2}(r_1,r_2|s_1,s_2)$ where $r_1,s_1\in I_1$ and $r_2,s_2\in I_2$, where
\begin{eqnarray}\label{v2}
r_1,s_1 \in J_1\;\;{\rm and}\;\;r_2,s_2 \in J_2 \;\;&\rightarrow &\;\;
[A_{m_1,m_2}]_{n_1,n_2}(r_1,r_2|s_1,s_2)=A_{m_1,m_2}(r_1,r_2|s_1,s_2)\nonumber\\
{\rm otherwise}\;\;&\rightarrow &\;\;[A_{m_1,m_2}]_{n_1,n_2}(r_1,r_2|s_1,s_2)=0
\end{eqnarray}
We use the notation ${\mathfrak L}_{m_1,n_1;m_2,n_2}(A_{m_1,m_2})$ for the tensor:
\begin{eqnarray}\label{v1}
{\mathfrak L}_{m_1,n_1;m_2,n_2}(A_{m_1,m_2})[\tau _{n_1,m_1}(r_1),\tau_{n_2,m_2}(r_2)|\tau _{n_1,m_1}(s_1),\tau_{n_2,m_2}(s_2)]=[A_{m_1,m_2}]_{n_1,n_2}(r_1,r_2|s_1,s_2)
\end{eqnarray}
The tensors ${\mathfrak L}_{m_1,n_1;m_2,n_2}(A_{m_1,m_2})$ and $A_{m_1,m_2}$ contain the same information.
Since we will use these tensors in a quantum mechanical context, we will say that 
the indices $r_1,s_1$ describe the first component and the indices $r_2,s_2$ the second component of this
`bipartite tensor'.
The partial trace of such a tensor with respect to the second component is the following $n_1\times n_1$ matrix:
\begin{eqnarray}
{\rm Tr}_2 [{\mathfrak L}_{m_1,n_1;m_2,n_2}(A_{m_1,m_2})]=\sum _{u=0}^{n_2-1}{\mathfrak L}_{m_1,n_1;m_2,n_2}
(A_{m_1,m_2})[r_1,u|s_1,u]
\end{eqnarray}
The partial transpose  of $A_{m_1,m_2}$ with respect to the second component is
\begin{eqnarray}
[A_{m_1,m_2}(r_1,r_2|s_1,s_2)]^{T_2}=A_{m_1,m_2}(r_1,s_2|s_1,r_2)
\end{eqnarray}
\begin{lemma}\label{500}
For $(m_1,m_2)|(n_1,n_2)$
\begin{itemize}
\item[(1)]
The $n_1n_2$ eigenvalues of ${\mathfrak L}_{m_1,n_1;m_2,n_2}(A_{m_1,m_2})$
are the $m_1m_2$ eigenvalues of 
$A_{m_1,m_2}$ plus $n_1n_2-m_1m_2$ zeros. 
\item[(2)]
\begin{eqnarray}\label{19}
{\rm Tr}_2 [{\mathfrak L}_{m_1,n_1;m_2,n_2}(A_{m_1,m_2})]={\mathfrak I}_{m_1,n_1}[{\rm Tr}_2 A_{m_1,m_2}]
\end{eqnarray}
Similar result is true, for the partial trace of the tensor with respect to the first component.
\item[(3)]
If $A_{m_1,m_2}=B_{m_1}\otimes C_{m_2}$ then
\begin{eqnarray}\label{191}
{\mathfrak L}_{m_1,n_1;m_2,n_2}(A_{m_1,m_2})={\mathfrak I}_{m_1,n_1}(B_{m_1})\otimes 
{\mathfrak I}_{m_2,n_2}(C_{m_2})
\end{eqnarray}

\end{itemize}
\end{lemma}
\begin{proof}
\mbox{}
\begin{itemize}
\item[(1)]
Using a bijective map between
${\mathbb Z}(m_1)\times {\mathbb Z}(m_2)$ and ${\mathbb Z}(m)$ where $m=m_1m_2$
we relabel the elements of $A_{m_1,m_2}$ and regard them as elements of the $m\times m$ matrix
$A'_{m}$. We do a similar relabeling for ${\mathfrak L}_{m_1,n_1;m_2,n_2}(A_{m_1,m_2})$ 
and then use lemma \ref{AA}.
\item[(2)]
From Eq.(\ref{v1}) it follows that
\begin{eqnarray}
{\rm Tr}_2 [{\mathfrak L}_{m_1,n_1;m_2,n_2}(A_{m_1,m_2})]={\rm Tr}_2[A_{m_1,m_2}]_{n_1,n_2}
\end{eqnarray}
From Eq.(\ref{v2}) it follows that
\begin{eqnarray}
{\rm Tr}_2[A_{m_1,m_2}]_{n_1,n_2}={\mathfrak I}_{m_1,n_1}[{\rm Tr}_2 A_{m_1,m_2}]
\end{eqnarray}
Combining these two equations, we prove Eq.(\ref{19}).

\item[(3)]
If $A_{m_1,m_2}=B_{m_1}\otimes C_{m_2}$ then 
\begin{eqnarray}
[A_{m_1,m_2}]_{n_1,n_2}(r_1,r_2|s_1,s_2)=[B_{m_1}]_{n_1}(r_1,s_1)[C_{m_2}]_{n_2}(r_2,s_2)
\end{eqnarray}
But
\begin{eqnarray}
&&[A_{m_1,m_2}]_{n_1,n_2}(r_1,r_2|s_1,s_2)={\mathfrak L}_{m_1,n_1;m_2,n_2}(A_{m_1,m_2})[\tau _{n_1,m_1}(r_1),\tau_{n_2,m_2}(r_2)|\tau _{n_1,m_1}(s_1),\tau_{n_2,m_2}(s_2)]\nonumber\\
&&[B_m]_n(r_1,s_1)={\mathfrak I}_{mn}(B_m)[\tau _{n,m}(r_1),\tau_{n,m}(s_1)]\nonumber\\
&&[C_m]_n(r_2,s_2)={\mathfrak I}_{mn}(C_m)[\tau _{n,m}(r_2),\tau_{n,m}(s_2)].
\end{eqnarray}
Therefore
\begin{eqnarray}
&&{\mathfrak L}_{m_1,n_1;m_2,n_2}(A_{m_1,m_2})[\tau _{n_1,m_1}(r_1),\tau_{n_2,m_2}(r_2)|\tau _{n_1,m_1}(s_1),\tau_{n_2,m_2}(s_2)]\nonumber\\&&=
{\mathfrak I}_{mn}(B_m)[\tau _{n,m}(r_1),\tau_{n,m}(s_1)]{\mathfrak I}_{mn}(C_m)[\tau _{n,m}(r_2),\tau_{n,m}(s_2)].
\end{eqnarray}
and this proves Eq.(\ref{191}).
\end{itemize}
\end{proof}

\section{The quantum system $\Sigma (n)$ with variables in ${\mathbb Z}(n)$}\label{NN}

We consider a quantum system $\Sigma (n)$ with positions and momenta in ${\mathbb Z}(n)$.
The Hilbert space $H(n)$ for this system is $n$-dimensional, and let
$|X_n;r\rangle$ where $r\in {\mathbb Z}(n)$, be an orthonormal basis that we call `basis of position states'.
The $X_n$ in the notation is not a variable, but it simply indicates that they are position states 
in the Hilbert space $H(n)$.
Through a Fourier transform we get another orthonormal basis that we call `basis of momentum states':
\begin{equation}
|{P_n};r\rangle=F_n|{X_n};r\rangle;\;\;\;\;
F_n=n^{-1/2}\sum _{r,s}\omega _n(rs)\ket{X_n;r}\bra{X_n;s}.
\end{equation}

The position-momentum phase space is the toroidal 
lattice ${\mathbb Z}(n)\times {\mathbb Z}(n)$.
Displacements in this phase space are performed with the operators
\begin{eqnarray}
&&{D}_n(\alpha, \beta, \gamma)=Z_n(\alpha) \;X_n(\beta) \;\omega _n(\gamma);\;\;\;\;
\alpha, \beta, \gamma\in {\mathbb Z}(n)
\nonumber\\
&&Z_n(\alpha) =\sum _{r}\omega _n(r \alpha)|{X_n};r\rangle \langle {X_n};r|=\sum _{r}|{P_n};r+\alpha\rangle \langle {P_n};r|
\nonumber\\
&&X_n(\beta)=
\sum _{r}\omega _n(-r \beta)|{P_n};r\rangle \langle {P_n};r|=\sum_r|{X_n};r+\beta\rangle \langle {X_n};r|
\end{eqnarray}
where
\begin{eqnarray}\label{546}
&&{D}_n(\alpha _1, \beta _1, \gamma _1){D}_n(\alpha _2, \beta _2, \gamma _2)=
{D}_n(\alpha _1+\alpha _2, \beta _1+\beta_2, \gamma _1+\gamma _2-\alpha _2\beta _1)\nonumber\\
&&[{D}_n(\alpha , \beta , \gamma )]^\dagger ={D}_n(-\alpha , -\beta , -\gamma-\alpha \beta )
\end{eqnarray}
Therefore the  operators ${D}_n(\alpha, \beta, \gamma)$ form a representation of the Heisenberg-Weyl group 
$HW[{\mathbb Z}(n)]$.

The symplectic operators $S_n(\kappa ,\lambda|\mu, \nu)$ in the present context have been discussed in 
detail in ref\cite{FI1}.
Here we only mention that 
\begin{eqnarray}
S_n(\kappa ,\lambda|\mu, \nu)X_n(1)[S_n(\kappa ,\lambda|\mu, \nu)]^\dagger &=& D_n(\lambda, \kappa,0)\nonumber\\
S_n(\kappa ,\lambda|\mu, \nu)Z_n(1)[S_n(\kappa ,\lambda|\mu, \nu)]^\dagger &=& D_n(\nu,\mu,0).
\end{eqnarray}
From this follows that $\kappa, \mu$ are associated with displacements in positions and 
$\lambda, \nu$ are associated with displacements in momenta. 
Therefore $\kappa, \mu$ take values in ${\mathbb Z}(n)$ and $\lambda, \nu$
take values in its Pontryagin dual group which is isomorphic to ${\mathbb Z}(n)$.
This has been used earlier, in section \ref {HW}. 

An arbitrary state $\ket{f_n}$ in $H(n)$ can be written as
\begin{eqnarray}
\ket{f_n}=\sum _{r=0}^{n-1} f_n(r)\ket{X_n;r}.
\end{eqnarray}
We call $R(n)$ the set of the density matrices $\rho _n$ of this system.

\section{Subsystems and supersystems: Embeddings and their compatibility}

We consider the systems $\Sigma(n)$ with $n\in {\mathbb X}$, and in each of them we 
consider an orthonormal basis $\ket{X_n;r}$ (where $r\in {\mathbb Z}(n)$).

\begin{definition}\label{def}
For $m|n$, the $\Sigma (m)$ is a subsystem of $\Sigma (n)$ (which we denote as $\Sigma (m)\prec \Sigma (n)$), or 
equivalently the $\Sigma (n)$ is a supersystem of $\Sigma (m)$ ($\Sigma (n)\succ \Sigma (m)$),
in the following sense:
\begin{itemize}
\item[(1)]
Position and momentum in the system $\Sigma (m)$ take values in ${\mathbb Z}(m)$, 
which is an additive subgroup of 
${\mathbb Z}(n)$, where the position and momentum of $\Sigma (n)$ belong.
\item[(2)]
There are several embeddings between these systems 
(which involve quantum states, density matrices, operators, etc) which preserve the structure and 
which are compatible to each other, in the sense 
discussed in the subsections below.
\end{itemize}
\end{definition}
The system $\Sigma (n)$ has $\sigma _0(n)-2$ `proper' subsystems. We exclude here itself as a subsystem, 
and a $1$-dimensional subsystem (as discussed in remark \ref{567}).
The set $\Sigma=\{\Sigma (2), \Sigma (3), ...\}$ with the partial order $\prec$, is a directed poset.

\subsection{Embeddings of quantum states}
For $m|n$, the Hilbert space $H(m)$ is embedded into $H(n)$ with the following linear map which is an injection:
\begin{eqnarray}\label{1}
&&{\cal A}_{mn}:\;\;
\sum _{r=0}^{m-1}f_m(r)\ket{X_m;r}\;\rightarrow\;
\sum _{s=0}^{n-1}f_n(s)\ket{X_n;s}\nonumber\\
&&s=dr\;\rightarrow\;f_n(s)=f_m\left (\frac{s}{d}\right );\;\;\;\;\;d=\frac{n}{m}\nonumber\\
&&{\rm otherwise}\;\rightarrow\;f_n(s)=0
\end{eqnarray}
The same map can also be written in terms of momentum states as
\begin{eqnarray}\label{1V}
&&{\cal A}_{mn}:\;\;
\sum _{r=0}^{m-1}g_m(r)\ket{P_m;r}\;\rightarrow\;
\sum _{s=0}^{n-1}g_n(s)\ket{P_n;s}\nonumber\\
&&s=r\;({\rm mod}\;m)\;\;\rightarrow\;g_n(s)=d^{-1/2}g_m(r)
\end{eqnarray}
The equivalence of Eqs(\ref{1}),(\ref{1V}) is proved with a Fourier transform.
According to Eq.(\ref{1V}), in the momentum basis, a vector $v$ in an $m$-dimensional space, is mapped to the vector
$d^{-1/2}(v,v,...,v)^T$ (where $v$ is repeated $d$ times) in the $n$-dimensional space.

\begin{proposition}
The map ${\cal A}_{mn}$ preserves the scalar product.
\end{proposition}
\begin{proof}
This follows immediately from Eq.(\ref{1}). 
\end{proof}
The various ${\cal A}_{mn}$ maps are compatible to each other, in the sense that
for all $m,n,\ell$ such that $m|n|\ell$, and for all vectors  
$\ket{f_m}$ in $H(m)$, we get ${\cal A}_{n\ell}[{\cal A}_{mn}(\ket{f_m})]={\cal A}_{m\ell}(\ket{f_m})$.
We use the following simplified notation to express this:
\begin{eqnarray}
m|n|\ell\;\rightarrow\;{\cal A}_{n\ell}\circ {\cal A}_{mn}={\cal A}_{m\ell}.
\end{eqnarray}
Below we use analogous simplified notation.

The adjoint map ${\cal A}_{mn}^\dagger$ to ${\cal A}_{mn}$ is a map from 
the dual Hilbert space $H^*(n)$ into the  dual Hilbert space $H^*(m)$ (where $m|n$) as follows:
\begin{eqnarray}\label{11}
&&{\cal A}_{mn}^\dagger:\;\;
\sum _{s=0}^{n-1}f_n(s)\bra{X_n;s}\;\rightarrow\;
\sum _{r=0}^{m-1}f_m(r)\bra{X_m;r}\nonumber\\
&&f_m(r)=f_n(dr)
\end{eqnarray}
The $\bra{X_n;s}$ with $s\ne dr$ in $H^*(n)$ are mapped into the zero vector in $H^*(m)$.

\begin{remark}\label{xc}
For every state $\ket{X_n;s}$ in $H(n)$ with $s\in {\mathbb Z}(n)-{\mathbb Z}^*(n)$ (where $n\notin \Pi$),
there exists another state $\ket{X_m;\frac{s}{d}}$ in $H(m)$, such that
\begin{eqnarray}
{\cal A}_{mn}\left (\ket{X_m;\frac{s}{d}}\right )=\ket{X_n;s};\;\;\;\;m=\frac{n}{d};\;\;\;\;
d={\rm GCD}(s,n)>1
\end{eqnarray}
Therefore the Hilbert space $H(n)$ can be partitioned into two parts
\begin{eqnarray}\label{2}
H(n)&=&H_A(n)\oplus H_B(n)\nonumber\\
H_A(n)&=&{\rm span}\{\ket{X_n;s}\;|\;s\in {\mathbb Z}(n)-{\mathbb Z}^*(n)\}\nonumber\\
H_B(n)&=&{\rm span}\{\ket{X_n;s}\;|\;s\in {\mathbb Z}^*(n)\}
\end{eqnarray}
The dimensions of $H_A(n), H_B(n)$ are $n-\varphi(n)$ and $\varphi(n)$, correspondingly.
$\Sigma (n)$ `shares' all states in $H_A(n)$, with some of its $\sigma_0(n)-2$ subsystems.
Physically, the information contained in the quantum states in the subsystems is also contained in the 
quantum states in $H_A(n)$.
Quantum states in $H_B(n)$, contain information which cannot be found in the subsystems. 

\end{remark}

\subsection{Embeddings of density matrices}

\begin{notation}\label{25}
If $\Theta _m$ is an operator (or a density matrix) we use the `tilde notation' for the $m\times m$ matrix
\begin{eqnarray}
\widetilde \Theta _m(r_1,r_2)=\bra{X_m;r_1}\Theta _m\ket{X_m;r_2},
\end{eqnarray} 
that consists of its matrix elements in the position basis $\ket{X_m;r}$.
We also use the notation
\begin{eqnarray}
{\mathfrak I}_{mn}(\Theta _m)=\sum _{r_1,r_2\in {\mathbb Z}(n)}[{\mathfrak I}_{mn}(\widetilde\Theta _m)](r_1,r_2)
\ket{X_n;r_1}\bra{X_n;r_2};\;\;\;\;m|n.
\end{eqnarray}
Analogous notation is also used for bipartite systems, below.
\end{notation}
It is easily seen that in the momentum basis
\begin{eqnarray}
s_1=r_1\;({\rm mod}\;m)\;{\rm and}\;s_2=r_2\;({\rm mod}\;m)\;\rightarrow\;
\bra{P_n;s_1}{\mathfrak I}_{mn}(\Theta _m)\ket{P_n;s_2}=\frac{1}{d}\bra{P_m;r_1}\Theta _m\ket{P_m;r_2}
\end{eqnarray}
where $d=n/m$. 
This is the analogue of Eq.(\ref{1V}) for density matrices. Here an $m\times m$ matrix $A$ is mapped
into the $n\times n$ matrix
\begin{eqnarray}
\frac{1}{d}
\begin{pmatrix}
A&A&\cdots&A\\
A&A&\cdots&A\\
\vdots &\vdots&\ddots&\vdots\\
A&A&\cdots&A
\end{pmatrix}
\end{eqnarray}
which contains $d^2$ times the matrix $A$.
Below we mainly work in the position basis, but everything can also be expressed in the momentum basis.

For $m|n$, the set of density matrices $R(m)$ is embedded into $R(n)$, with the injection
\begin{eqnarray}\label{2}
&&{\mathfrak I}_{mn}:\;\;\rho _m\;\rightarrow\;{\mathfrak I}_{mn}(\rho _m).
\end{eqnarray}
\begin{proposition}
The map ${\mathfrak I}_{mn}$ preserves the trace of the product of two matrices.
\end{proposition}
\begin{proof} 
Indeed, using lemma \ref{AA}, we prove that
\begin{eqnarray}\label{20}
{\rm Tr}(\rho _m \rho' _m)={\rm Tr}[{\mathfrak I}_{mn}( \rho _m){\mathfrak I}_{mn}( \rho' _m)].
\end{eqnarray}
\end{proof}
A consequence of this is the following:
\begin{cor}
The map ${\mathfrak I}_{mn}$ preserves the ${\rm Tr}( \rho _m ^2)$, which can be regarded
as a measure of how mixed (or pure) the state is.
\end{cor}

The eigenvalues of $\widetilde\rho_n$ (where the tilde notation is defined in \ref{25}), 
are the eigenvalues of $\widetilde\rho_m$, plus $n-m$ zeros.
Also, if $\ket{v}$ is an eigenvector of $\widetilde \rho_m$, then ${\cal A}_{mn}[\ket{v}]$
is an eigenvector of $\widetilde\rho _n$. A consequence of this is the following compatibility 
condition between the
${\mathfrak I}_{mn}$ and ${\cal A}_{mn}$:
\begin{eqnarray}
\rho _m=\sum _{r=0}^{m-1} p_r\ket{v_r}\bra{v_r}\;\rightarrow\;
{\mathfrak I}_{mn}[\rho _m]=\sum _{r=0}^{m-1} p_r\ket{u_r}\bra{u_r};\;\;\;\;\ket{u_r}={\cal A}_{mn}[\ket{v_r}]
\end{eqnarray}
In addition to that, the various ${\mathfrak I}_{mn}$ maps are compatible to each other:  
\begin{eqnarray}
m|n|\ell\;\rightarrow\;{\mathfrak I}_{n\ell}\circ {\mathfrak I}_{mn}={\mathfrak I}_{m\ell}
\end{eqnarray}

\subsection{Embeddings of orthogonal projectors}\label{proje}
A measurement on $\Sigma (n)$ is described by an $n$-tuple 
of orthogonal projectors acting on $H(n)$:
\begin{eqnarray}
\pi _s(n)\pi _q(n)=\pi _s(n)\delta(s,q);\;\;\;\;\;s=0,...,(n-1).
\end{eqnarray}
For $m|n$, the set of all $m$-tuples of projectors is embedded into the set 
of all $n$-tuples of projectors as follows.
\begin{eqnarray}
&&{\mathfrak P}_{mn}:\;\;\{\pi _r(m)\;|\;r=0,...,m-1\}\;\;\rightarrow\;\;\{\pi _{s}(n)\;|\;s=0,...,n-1\}\nonumber\\
&&\pi _{dr}(n)={\mathfrak I}_{mn}(\pi _{r}(m));\;\;\;\;\;
d=\frac{n}{m};\;\;\;\;\;r=0,...,m-1.
\end{eqnarray}
This defines $m$ of the projectors in the $n$-tuple (those with index $dr$), and the rest are chosen so that 
together with the $\pi _{dr}(n)$
they form an orthogonal set of $n$ projectors.
There are many ways of doing this, but the results below do not depend on the particular choice.

We note that the outcome $s$ associated with a measurement on $\Sigma (m)$ described by $\pi _s(m)$,
corresponds to the outcome $ds$ associated with a measurement on $\Sigma (n)$ described by $\pi _{ds}(n)$.
Also we can prove that for $m|n|\ell$ we get ${\mathfrak P}_{n\ell}\circ {\mathfrak P}_{mn}={\mathfrak P}_{m\ell}$.

These projectors are compatible with the embedding of Eq.(\ref{2}) in the sense that for $m|n$
\begin{eqnarray}
&&\pi _{dr}(n){\mathfrak I}_{mn}(\rho _m)={\mathfrak I}_{mn}[\pi _r(m)\rho _m]
;\;\;\;r=0,...,m-1\nonumber\\
&&\pi _{s}(n){\mathfrak I}_{mn}(\rho _m)=0;\;\;\;s\ne dr
\end{eqnarray}
A measurement with the projectors $\pi _r(m)$ on the density matrix $\rho _m$ in the 
system $\Sigma (m)$, will give the result $r$ with probability ${\rm Tr}[\pi _r(m)\rho _m]$.
A measurement with the projectors $\pi _{s}(n)$ on the density matrix ${\mathfrak I}_{mn}(\rho _m)$ in the 
system $\Sigma (n)$, will give the corresponding result $dr$ with the same probability 
${\rm Tr}[\pi _{dr}(n){\mathfrak I}_{mn}(\rho _m)]={\rm Tr}[\pi _r(m)\rho _m]$, and it will never give outcome
different than $dr$.

A density matrix $\rho _m$ after a non-selective measurement described by the projectors $\{\pi _s(m)\}$ becomes 
\begin{eqnarray}\label{bnz}
{\cal M}(\rho_m)=\sum _{r=0}^{m-1}\pi _r(m)\rho _m\pi _r(m).
\end{eqnarray}
This formalism is compatible with 
the embeddings ${\mathfrak I}_{mn}$, in the sense that
\begin{eqnarray}
m|n\;\;\rightarrow\;\;{\mathfrak I}_{mn}\left [\sum _{r=0}^{m-1}\pi _r(m)\rho _m\pi _r(m)\right ]
=\sum _{s=0}^{n-1}\pi _s(n)\;{\mathfrak I}_{mn}(\rho _m)\;\pi _s(n)
\end{eqnarray}
In the sum on the right hand side, all terms corresponding to $s\ne dr$ are equal to zero.
This is consistent with our comment earlier that although there are many ways to choose $\pi _{s}(n)$ 
with $s\ne dr$, 
the result does not depend on the particular choice.  

\subsection{Embeddings of displacement and symplectic operators}

For $m|n$, the embeddings ${\mathfrak X}_{mn}$ 
of the displacecement operators are similar to the one discussed in section \ref{HW}
(although here we have a different representation of the Heisenberg-Weyl group). 
We note here that these embeddings are compatible with the ${\cal A}_{mn}$, in the sense that 
\begin{eqnarray}
m|n\;\rightarrow\;{\cal A}_{mn}\circ D_m(\alpha,\beta, \gamma)={\mathfrak X}_{mn}[D_m(\alpha,\beta, \gamma)]\circ  {\cal A}_{mn};\;\;\;\;
d=\frac{n}{m}.
\end{eqnarray}
Similar statement can be made for the symplectic operators:
\begin{eqnarray}
m|n\;\rightarrow\;{\cal A}_{mn}\circ S_m(\kappa, \lambda |\mu,\nu)
={\mathfrak X}_{mn}[S_m(\kappa, \lambda |\mu,\nu)]\circ  {\cal A}_{mn};\;\;\;\;
d=\frac{n}{m}.
\end{eqnarray}

\section{Ubiquitous and nonubiquitous quantities}\label{ubi}

In each system $\Sigma (n)$ we can define various quantities. Below we give three different categories of such quantities,
which are formally defined as maps
\begin{eqnarray}\label{bb}
&&E_n:\;\;H(n)\;\rightarrow \;{\mathbb R};\;\;\;E_n(\ket{f_n})\in {\mathbb R}\nonumber\\
&&E_n:\;\;H(n)\;\rightarrow \;{\mathbb M}_n;\;\;\;E_n(\ket{f_n})\in {\mathbb M}(n)\nonumber\\
&&E_n:\;\;{\mathbb R}\;\rightarrow \;R(n);\;\;\;\;E_n(\lambda)\in R(n)
\end{eqnarray}
${\mathbb M}(n)$ is the set of $n\times n$ complex matrices.
Entropic quantities are examples in the first category.
Wigner and Weyl functions are examples in the second category.
Density matrices $E_n$ given by 
\begin{eqnarray}\label{bb1}
\lambda\;\rightarrow \;E_n(\lambda)=\sum _{r=0}^{n-1}p_n(r;\lambda)\ket{X_n;r}\bra{X_n;r};\;\;\;\;\sum _{r=0}^{n-1}p_n(r;\lambda)=1,
\end{eqnarray}
which are associated with a family of distributions $p_n(r;\lambda)$
are examples in the third category. 
An example is the case where $p_n(r;\lambda)=\lambda ^r(\lambda -1)/(\lambda^n-1)$.
Such matrices may be useful in a particular application.

The following important questions arise:
\begin{itemize}
\item
If we calculate the quantity $E_m (\ket {f_m})$ for a state $\ket {f_m}$ of the system $\Sigma (m)$, 
and then regard this state within the supersystem 
$\Sigma (n)$ (for $m|n$) and we calculate $E_n[{\cal A}_{mn}(\ket {f_m})]$, will we get the same answer?
In other words, we explore whether the relation
\begin{eqnarray}\label{rf}
E_n[{\cal A}_{mn}(\ket {f_m})]=E_m (\ket {f_m})
\end{eqnarray}
is true, for all states $\ket {f_m}$.

In the case that $E_m (\ket {f_m})$ is an $m\times m$ matrix the term `same answer' means 
$E_n[{\cal A}_{mn}(\ket {f_m})]={\mathfrak I}_{mn}[E_m (\ket {f_m})]$.
Also for the quantities in the third category in Eq.(\ref{bb}), the analogous question is whether ${\mathfrak I}_{mn}[E_m(\lambda)]=E_n(\lambda)$.
\item
If there is a property among $E_m (\ket {f_m})$, do the $E_n[{\cal A}_{mn}(\ket {f_m})]$ have
an analogous property?
\end{itemize}
We call ubiquitous quantities the ones for which the answer is positive, and nonubiquitous the ones for 
which the answer is negative. Nonubiquitous quantities are `multivalued' in the sense that
for a state $\ket {f_m}$ in $H(m)$, we can consider the quantities
\begin{eqnarray}\label{601}
{\mathfrak S}(\ket{f_m})=\{E_n[{\cal A}_{mn}(\ket {f_m})]\;|\;n=2m,3m,...\}
\end{eqnarray}
which are in general different from each other.
Therefore ubiquity is a concept which tells us which quantities fit our structure of regarding smaller systems as subsystems of larger ones. 
Below we make formal these ideas.

\begin{definition}
The $\mathbb E=\{E_2,E_3,...\}$ is an ubiquitous quantity in $\Sigma =\{\Sigma(2), \Sigma (3),...\}$
if for $m|n$:
\begin{itemize}
\item[(1)]
$E_n\circ {\cal A}_{mn}=E_m$, for quantities in the first category in Eq.(\ref{bb}).
\item[(2)]
$E_n[{\cal A}_{mn}(\ket {f_m})]={\mathfrak I}_{mn}[E_m(\ket {f_m})]$, for quantities in the second category in Eq.(\ref{bb}).
\item[(3)]
${\mathfrak I}_{mn}[E_m(\lambda)]=E_n(\lambda)$, for quantities in the third category in Eq.(\ref{bb}).
\end{itemize}
In addition to that, if there is a property among $E_m (\ket {f_m})$ this should be preserved.
 
Analogous definition can be given for quantities which are defined on the set $R(n)$ of density matrices (in which case, compatibility 
with the embeddings ${\mathfrak I}_{mn}$ is required).
Furthermore, the definition can also be extended to bipartite systems.
\end{definition}

A related topic which we do not discuss in this paper is the question of ubiquitous Hamiltonians.
They are a set of  Hamiltonians $\{H_n\}$ (where $H_n$ corresponds to the system $\Sigma (n)$) such that
Eq.(\ref{rf}) is valid at all times $t$, i.e., 
\begin{eqnarray}
E_n[\exp (iH_n t){\cal A}_{mn}(\ket {f_m})]=E_m[\exp (iH_m t)\ket {f_m}]. 
\end{eqnarray}
Below we discuss several examples of ubiquitous quantities.
It is easily seen that for many families of distributions $p_n(r;\lambda)$ 
the quantities in the third category in Eq.(\ref{bb}), are nonubiquitous.

\subsection{Entropic functions}\label{entro}

In the system $\Sigma (n)$ we consider the entropy maps
\begin{eqnarray}
&&S_n:\;\;R(n)\;\rightarrow \;\;[0,\ln n]\subset {\mathbb R}_0^+
;\;\;\;\;S_n(\rho _n)=-{\rm Tr} (\rho_n\log \rho_n )\nonumber\\
&&S_n':\;\;R(n)\;\rightarrow \;\;[0,\ln n]\subset {\mathbb R}_0^+
;\;\;\;\;S_n'(\rho _n; \{\pi _r(n)\})=-{\rm Tr}[{\cal M}(\rho_n)\log {\cal M}(\rho_n)],
\end{eqnarray}
where ${\cal M}(\rho_n)$ is associated with a non-selective measurement with the
projectors $\pi _r(m)$, and has been given in Eq.(\ref{bnz}).
Let ${\mathbb S}=\{S_2,S_3,...\}$ and ${\mathbb S}'=\{S_2',S_3',...\}$.
We note that ${\mathbb S}'$ is defined with respect to a series of projectors
$\{\{\pi _{r_2}(2)\}, \{\pi _{r_3}(3)\},...\}$ which are compatible as discussed in section\ref{proje}.

\begin{proposition}\label{P1}
Both entropies ${\mathbb S}$ and ${\mathbb S}'$ 
are  ubiquitous quantities in $\Sigma$.
\end{proposition}
\begin{proof}
Both of these maps are compatible with the embeddings ${\mathfrak I}_{mn}$ (where $m|n$):
\begin{eqnarray}\label{296}
S_n \circ {\mathfrak I}_{mn}=S_m;\;\;\;\;S_n '\circ {\mathfrak I}_{mn}=S_m'
\end{eqnarray}
The proof of this is based on the fact that these embeddings preserve the eigenvalues, 
with extra eigenvalues which are zeros (lemma \ref{AA}).
Therefore both $\mathbb S$ and $\mathbb S'$ are ubiquitous quantities.
\end{proof}
\begin{remark}
We note that if we consider another embedding 
with respect to another basis (e.g., the basis of momentum states) Eqs.(\ref{296}) are still valid, because the
entropy depends on the eigenvalues of these matrices. 
\end{remark}

\subsection{Wigner and Weyl functions in systems with odd dimension}

In this section we show briefly that the Wigner and Weyl functions are ubiquitous quantities. 
 It is known that the Wigner and Weyl functions are slightly different in systems with even and odd dimension.
For this reason in this particular example we consider only systems with odd dimension.
So the set ${\mathbb X}$ in Eq.(\ref{ty}) is replaced with its subset
\begin{eqnarray}
{\mathbb X}_{\rm odd}=\{3,5,...\}
\end{eqnarray}
which can also become a topological space with the divisor topology (in a way analogous to ${\mathbb X}$).
In these systems it is convinient to work with a more `symmetric' definition of the displacement operators:
\begin{eqnarray}
&&{\mathfrak D}_n(\alpha, \beta, \gamma)=Z_n(\alpha) \;X_n(\beta) \;\omega _n(\gamma-2^{-1}\alpha \beta)
;\;\;\;\;
\alpha, \beta, \gamma\in {\mathbb Z}(n)
\nonumber\\
&&{\mathfrak D}_n(\alpha _1, \beta _1, \gamma _1){\mathfrak D}_n(\alpha _2, \beta _2, \gamma _2)=
{\mathfrak D}_n[\alpha _1+\alpha _2, \beta _1+\beta_2, \gamma _1+\gamma _2
+2^{-1}(\alpha _1\beta _2-\alpha _2\beta _1)]\nonumber\\
&&[{\mathfrak D}_n(\alpha , \beta , \gamma )]^\dagger ={\mathfrak D}_n(-\alpha , -\beta , -\gamma)
\end{eqnarray}
We note that in ${\mathbb Z}(n)$ with odd $n$, the $2^{-k}$ exists.
We also define the displaced parity operator\cite{FI1} 
\begin{eqnarray}
P_n(\alpha, \beta)&=&{\mathfrak D}_n(\alpha, \beta, 0)P_n 
[{\mathfrak D}_n(\alpha , \beta , 0)]^\dagger\nonumber\\
P_n\ket{X_n;r}&=&\ket{X_n;-r}
\end{eqnarray}
The Wigner function $W_n$ 
and the Weyl function $\widehat W_n$ are the following
maps from the set of operators $\Theta _n$ in the system $\Sigma (n)$, to the set of $n\times n$ matrices:
\begin{eqnarray}
W_n(\Theta _n;\alpha, \beta)={\rm Tr} [\Theta _n P_n(\alpha, \beta)];\;\;\;\;
\widehat W_n(\Theta _n;\alpha, \beta)={\rm Tr} [\Theta _n {\mathfrak D}_n(\alpha, \beta)];\;\;\;\;
\alpha, \beta \in {\mathbb Z}(n)
\end{eqnarray}
The star product gives the Wigner function of $\Theta _n \Phi _n$ in terms of the Wigner functions
of $\Theta _n$ and $ \Phi _n$, as follows: 
\begin{eqnarray}\label{Q1}
W_n(\Theta _n; \alpha , \beta )\star W_n(\Phi _n;\alpha , \beta )&\equiv & 
W_n(\Theta _n \Phi _n; \alpha , \beta )=\sum _{\alpha_1,\beta _1,\alpha_2,\beta _2}
W_n(\Theta _n; \alpha +\alpha _1, \beta +\beta _1)\nonumber\\&\times &
W_n(\Phi _n; \alpha +\alpha _2, \beta +\beta _2)
\omega _n[2(\alpha_2\beta _1-\alpha_1\beta _2)].
\end{eqnarray}
Analogous formula can be given for the Weyl functions:
\begin{eqnarray}\label{Q2}
\widehat W_n(\Theta _n \Phi _n; \alpha , \beta )&=&\sum _{\alpha ',\beta '}
\widehat W_n(\Theta _n; 2^{-1}\alpha +\alpha ', 2^{-1}\beta +\beta ')\nonumber\\&\times 
&\widehat W_n(\Phi _n; 
2^{-1}\alpha -\alpha ', 2^{-1}\beta -\beta ')
\omega _n(2^{-1}\alpha '\beta -2^{-1}\alpha \beta ').
\end{eqnarray}
The proof of Eqs.(\ref{Q1}),(\ref{Q2}), is lengthy but straightforward.

Let ${\mathbb W}=\{W_2,W_3,...\}$ and $\widehat {\mathbb W}=\{\widehat W_2,\widehat W_3,...\}$.
The following result is intimately linked with the embeddings of the Heisenberg-Weyl groups discussed in section
\ref{HW}.
\begin{proposition}\label{P2}
Both the Wigner function ${\mathbb W}$ and the Weyl function $\widehat {\mathbb W}$, 
are ubiquitous quantities.
\end{proposition}
\begin{proof}
For $m|n$ we consider the embeddings:
\begin{eqnarray}
{W}_m(\Theta _m;\alpha, \beta)\;\;\rightarrow\;\;{W}_n(\mathfrak I_{mn}(\Theta _m);d\alpha, \beta)
;\;\;\;\;d=\frac{n}{m}
\end{eqnarray}
It is easily seen, that we
can map the star product of two Wigner functions in ${\mathbb W}(m)$ into the corresponding star product in 
${\mathbb W}(n)$.
Therefore the set ${\mathfrak W}=\{{\mathbb W}(2), {\mathbb W}(3),...\}$ of 
Wigner functions in $\{\Sigma (n)\}$ (for all $n$), is not an ubiquitous quantity.

Analogous proof can be given for the Weyl function.
\end{proof}

\section{Topological spaces}
\subsection{The directed posets $\mathfrak H$ and $\mathfrak R$ as topological spaces  with the divisor topology}\label{topo}

Let 
\begin{eqnarray}
\mathfrak H&=&\{H(2),H(3),...\}\nonumber\\
\mathfrak R&=&\{R(2),R(3),...\}.
\end{eqnarray}
${\mathfrak H}$ is a directed poset with $H(m)\prec H(n)$ if and only if $m|n$.
In this case the partial order is `subsystem' (i.e., $H(m)$ describes a subsystem of the system described by $H(n)$).
The partial order can also be interpreted as a partial order of information, in the following sense.
All the information contained in the quantum states in $H(m)$ is also contained in the quantum states in $H(n)$ (for $m|n$).
In fact, the quantum states in $H(n)$ contain more information than those in $H(m)$ (see also remark \ref{xc}).
${\mathfrak H}$ is a directed poset, and this means that the information contained in the states in two spaces $H(r)$ and $H(s)$,
is also contained in another space in $\mathfrak H$, which actually is $H[{\rm LCM}(r,s)]$.
These intuitive physical concepts become formal with the partial order and topology.

The map 
\begin{eqnarray}
f:\;{\mathbb X}\;\rightarrow\;\mathfrak H;\;\;\;\;f(n)=H(n)
\end{eqnarray}
is a bijection and  ${\mathfrak H}$ is a topological space with the divisor topology ${\mathfrak T}_{\mathfrak H}$,
as discussed in remark \ref{rema}. The same is true for $\mathfrak R$.
The following topological spaces are homeomorphic to each other:
\begin{eqnarray}
({\mathbb X}, {\mathfrak T}_{\mathbb X})\sim ({\mathfrak Z}, {\mathfrak T}_{\mathfrak Z})\sim
(\mathfrak H, {\mathfrak T}_{\mathfrak H})\sim (\mathfrak R, {\mathfrak T}_{\mathfrak R}).
\end{eqnarray}
With respect to partial order also, ${\mathbb X}, {\mathfrak Z}, {\mathfrak H}, {\mathfrak R}$ are order isomorphic and 
proposition \ref{q2} holds for all of them.

\subsection{Physical meaning of the divisor topology}\label{QQ}

Many topologies can be defined on a particular set.
The divisor topology considered here reflects the physical concepts of subsystem and supersystem, 
introduced earlier.

An open (resp. closed) set includes some systems and all their subsystems (resp. supersystems).
As an example, we consider the open set
\begin{eqnarray}
U_{\mathfrak H}(6)\cup U_{\mathfrak H}(15)=\{H(2),H(3),H(5),H(6),H(15)\}
\end{eqnarray}
This has been considered earlier in example \ref{315}, and here it is simply expressed in the context 
of Hilbert spaces.
This contains the Hilbert spaces of the systems $\Sigma (6)$ and $\Sigma (15)$ and also the Hilbert spaces 
of their subsystems
$\Sigma (2), \Sigma (3), \Sigma (5)$. 

We also consider the closed set $\overline {\{H(3)\}}=\{H(3),H(6),H(9),...\}$
which is actually the closure of $\{H(3)\}$. 
This contains the Hilbert spaces of all the supersystems of $\Sigma (3)$.

We have seen in proposition \ref{propo}, 
that ${\mathfrak H}$ is a $T_0$-space, but it is not a $T_1$-space. This
reflects very fundamental aspects of the relationship between a finite system and its subsystems and 
supersystems. 
The fact that ${\mathfrak H}$ is a $T_0$-space,
means that for any distinct elements $H(n)$ and $H(m)$, there is an open set containing one 
of them but not the other. 
From a physical point of view, if one of the systems is a subsystem of the other, 
e.g., $\Sigma (m)\prec \Sigma (n)$,
then $\{H(\ell)\;|\;\ell|m\}$ is an open set which contains $H(m)$ and not $H(n)$.
On the other hand if none of the $\Sigma (m), \Sigma (n)$ is a subsystem of the other,  
then again $\{H(\ell)\;|\;\ell|m\}$ is an open set which contains $H(m)$ and not $H(n)$.

On the other hand ${\mathfrak H}$ is not a $T_1$-space. In a $T_1$-space, for any pair of
elements ${H}(m)$ and ${H}(n)$, there exist two open sets $U_{\mathfrak H}$ and $U'_{\mathfrak H}$ such that
\begin{eqnarray}
{H}(n) \in U_{\mathfrak H};\;\;\;\;\;{H}(n) \notin U'_{\mathfrak H};\;\;\;\;\;
{H}(m) \in U'_{\mathfrak H};\;\;\;\;\;{H}(m) \notin U_{\mathfrak H}
\end{eqnarray}
But this is impossible if ${H}(m)\prec {H}(n)$ and therefore ${\mathfrak H}$ is not a $T_1$-space. 
It is seen that the properties of the topology reflect very fundamental logical
relationships between a system and its subsystems.

Let $A$ be a subset of ${\mathfrak H}$.
A neighborhood of $A$ is a subset of ${\mathfrak H}$ which contains an open set containing $A$, i.e.,
it contains spaces of subsystems and supersystems
of those in $A$. For example, if $A=\{H(6),H(9)\}$ the $\{H(2),H(3),H(6),H(9),H(18)\}$ is a neighborhood of $A$.  
Physically, it is natural to have the spaces $H(2), H(3)$ of the subsystems $\Sigma (2), \Sigma (3)$ and the space
$H(18)$ of the supersystem $\Sigma (18)$ in a neighborhood of $A$ and 
this is precisely what the divisor topology does.
The physical concepts of subsystems and supersystems become formal with the topology formalism.
We stress that topology, defines neighborhoods, continuity, etc, without the concept of distance.

\subsection{Topological spaces of ubiquitous quantities}

The set ${\mathbb E}$ defined in section \ref{ubi} is a directed poset with 
$E_m\prec E_n$ if and only if $m|n$.
In this case the partial order is `quantity in subsystem' (i.e., $E_n$ is a quantity in a system and $E_m$ is the corresponding quantity in a subsystem ).
The map 
\begin{eqnarray}
f:\;{\mathbb X}\;\rightarrow\;\mathbb E;\;\;\;\;f(n)=E_n
\end{eqnarray}
is a continuous bijection and $\mathbb E$ is a topological space with the divisor topology 
${\mathfrak T}_{\mathbb E}$,
as discussed in remark \ref{rema}. Ubiquity is important for the subsystem interpretation 
of the partial order in ${\mathbb E}$.
If $\mathbb E$ is not an ubiquitous quantity, we have a `multivaluedness', where
to a state $\ket{f_m}$ we attach one of the quantities in the set 
${\mathfrak S}(\ket{f_m})$ in (\ref{601}), depending on which of the supersystems $\Sigma (mN)$
we embed this state.
Therefore, $(\mathbb E, {\mathfrak T}_{\mathbb E})$ fits the spirit of this paper, 
only in the case of ubiquitous quantities.

As an example we consider the entropy $\mathbb S$ (defined in section \ref{entro}) which is a directed poset with $S_m\prec S_n$ if and only if $m|n$. 
$\mathbb S$ is order isomorphic to ${\mathbb X}$.
Since  $\mathbb S$ is an ubiquitous quantity, 
we can give the subsystem interpretation to $S_m\prec S_n$.
Then we can make $\mathbb S$ a topological space with the divisor topology 
${\mathfrak T}_{\mathbb S}$ (as in remark \ref{rema}), and the map ${\mathbb X}\;\rightarrow \;{\mathbb S}$
is continuous.

\section{Bipartite systems }

\subsection{Subsystems and supersystems }

$\Sigma (n_1,n_2)$ is a bipartite system described by the Hilbert space $H(n_1)\otimes H(n_2)$.
We call $R(n_1,n_2)$ the set of the density matrices $\rho_{n_1n_2}$ of this system.
The concepts of subsystems and supersystems discussed earlier, can be extended to bipartite systems.
For $(m_1,m_2)|(n_1,n_2)$ the system  $\Sigma (m_1,m_2)$ is a subsystem of $\Sigma (n_1,n_2)$
in the sense of properties analogous to those discussed earlier.
Some parts of this generalization are straightforward and we do not repeat all the technical details again, but we discuss the compatibility between the 
formalism earlier and the one for bipartite systems.  

For $(m_1,m_2)|(n_1,n_2)$, the embedding of $H(m_1)\otimes H(m_2)$ into 
$H(n_1)\otimes H(n_2)$, is the linear map
\begin{eqnarray}\label{1a}
&&{\cal B}_{m_1n_1;m_2n_2}:\;\;
\sum _{r_1,r_2}f_{m_1m_2}(r_1,r_2)\ket{X_{m_1};r_1}\otimes \ket{X_{m_2};r_2}\;\rightarrow\;
\sum _{s_1,s_2}f_{n_1n_2}(s_1,s_2)\ket{X_{n_1};s_1}\otimes \ket{X_{n_2};r_2}\nonumber\\
&&s_1=d_1r_1\;\;{\rm and}\;\;s_2=d_2r_2\;\rightarrow\;f_{n_1n_2}(s_1,s_2)=f_{m_1m_2}\left (\frac{s_1}{d_1},\frac{s_2}{d_2}\right );\;\;\;\;\;d_i=\frac{n_i}{m_i};\;\;\;\;i=1,2\nonumber\\
&&{\rm otherwise}\;\rightarrow\;f_{n_1n_2}(s_1,s_2)=0.
\end{eqnarray}
Also the set of density matrices $R(m_1,m_2)$ is embedded into $R(n_1,n_2)$, as follows:
\begin{eqnarray}\label{29}
{\mathfrak L}_{m_1,n_1;m_2,n_2}
:\;\;\rho _{m_1,m_2}\;\rightarrow\;\rho _{n_1,n_2}=
{\mathfrak L}_{m_1,n_1;m_2,n_2}(\rho _{m_1,m_2})
\end{eqnarray}

\begin{proposition}
For $(m_1,m_2)|(n_1,n_2)$,
\begin{itemize}
\item[(1)]
\begin{eqnarray}\label{27}
{\rm Tr}_2[{\mathfrak L}_{m_1n_1;m_2n_2}(\rho _{m_1,m_2})]={\mathfrak I}_{m_1n_1}[{\rm Tr}_2(\rho _{m_1,m_2})]
\end{eqnarray}
Similar result holds for the trace with respect to the first component system.
\item[(2)]
For a separable density matrix
\begin{eqnarray}
{\mathfrak L}_{m_1,n_1;m_2,n_2}\left (\sum _ip_i\rho^{(i)}_{m_1}\otimes \rho^{(i)}_{m_2}\right )
=\sum _ip_i{\mathfrak I}_{m_1,n_1}(\rho^{(i)}_{m_1})
\otimes {\mathfrak I}_{m_2,n_2}(\rho^{(i)}_{m_2})
\end{eqnarray}
where $p_i$ are probabilities.
A special case of this is that for a factorizable density matrix
\begin{eqnarray}
{\mathfrak L}_{m_1,n_1;m_2,n_2}(\rho_{m_1}\otimes \rho_{m_2})={\mathfrak I}_{m_1,n_1}(\rho_{m_1})
\otimes {\mathfrak I}_{m_2,n_2}(\rho_{m_2})
\end{eqnarray}
Therefore the map ${\mathfrak L}_{m_1,n_1;m_2,n_2}$ preserves the factorizable, separable 
or entangled nature of the density matrix.
\end{itemize}
\end{proposition}
\begin{proof}
The proof of both statements is based on lemma \ref{500} (parts 2 and 3, correspondingly).
\end{proof}

\subsection{Topological spaces}\label{topo1}

Let 
\begin{eqnarray}
{\mathfrak H}_2&=&\{H(n_1)\otimes H(n_2)\;|\;n_1,n_2=2,3...\}\nonumber\\
{\mathfrak R}_2&=&\{R(n_1,n_2)\;|\;n_1,n_2=2,3...\}
\end{eqnarray}
${\mathfrak H}_2$ is a directed poset with partial order $H(m_1)\otimes H(m_2)\prec H(n_1)\otimes H(n_2)$ if and only if 
$(m_1,m_2)|(n_1,n_2)$.
This is the `subsystem' partial order.
The map 
\begin{eqnarray}
f:\;{\mathbb X}\times {\mathbb X}\;\rightarrow\;{\mathfrak H}_2;\;\;\;\;f(n_1,n_2)=H(n_1,n_2)
\end{eqnarray}
is a continuous bijection and $\mathfrak H_2$ is a topological space with the divisor topology 
${\mathfrak T}_{\mathfrak H _2}$,
as discussed in remark \ref{rema}. The same is true for $\mathfrak R _2$.
The following 
topological spaces are homeomorphic to each other:
\begin{eqnarray}
({\mathbb X}\times {\mathbb X}, {\mathfrak T}_{{\mathbb X}\times {\mathbb X}})\sim 
(\mathfrak H _2, {\mathfrak T}_{\mathfrak H _2})\sim (\mathfrak R _2, {\mathfrak T}_{\mathfrak R _2})
\end{eqnarray}
The directed posets ${\mathbb X}\times {\mathbb X}, \mathfrak H _2, \mathfrak R _2$ are order isomorphic.

\subsection{Ubiquity of entanglement quantities}

The concept of ubiquitous quantities, can also be extended to bipartite systems.
Below we discuss some examples of quantities used in studies of entanglement.
All these quantities are maps from $R(n_1,n_2)$ to ${\mathbb R}$.

In the system $\Sigma (n_1,n_2)$ we consider the entropy $S_{n_1,n_2}(\rho_{n_1,n_2})$
of the whole system and the entropies $S_{n_1}[{\rm Tr}_2(\rho_{n_1,n_2})]$ and 
$S_{n_2}[{\rm Tr}_1(\rho_{n_1,n_2})]$ of the two component systems.
We also consider the 
\begin{eqnarray}
&&I_{n_1,n_2}(\rho _{n_1,n_2})=S_{n_1}[{\rm Tr}_2(\rho _{n_1,n_2})]+S_{n_2}[{\rm Tr}_1(\rho _{n_1,n_2})]
-S_{n_1,n_2}(\rho _{n_1,n_2})\nonumber\\
&&I'_{n_1|n_2}(\rho _{n_1,n_2})=S_{n_1,n_2}(\rho _{n_1,n_2})-S_{n_2}[{\rm Tr}_1(\rho _{n_1,n_2})].
\end{eqnarray}
The first of these quantities is the quantum mutual information and quantifies the correlations between the 
two component systems. The second quantity (and also $I'_{n_2|n_1}(\rho _{n_1,n_2})$ 
which is defined in a similar way) is the conditional entropy and it can be used as an entanglement witness
\cite{AC}.
Another quantity is the negativity\cite{VW} which is defined as
\begin{eqnarray}
{\mathfrak N}_{n_1,n_2}(\rho_{n_1,n_2})=\frac{||\rho _{n_1,n_2}^{T_1}||-1}{2}
\end{eqnarray} 
where for a Hermitian operator $A$, the trace norm is given by $||A||={\rm Tr}(A^{\dagger}A)^{1/2}$. 
\begin{proposition}\label{P3}
The following are ubiquitous quantities in 
$\{\Sigma (n_1,n_2)\;|\;(n_1,n_2) \in {\mathbb X}\times {\mathbb X}\}$:
\begin{itemize}
\item[(1)]
The entropy $\{S_{n_1,n_2}\;|\;(n_1,n_2) \in {\mathbb X}\times {\mathbb X}\}$
\item[(2)]
the quantum mutual information 
$\{I_{n_1,n_2}\;|\;(n_1,n_2) \in {\mathbb X}\times {\mathbb X}\}$
\item[(3)]
the conditional entropy 
$\{I'_{n_1|n_2}\;|\;(n_1,n_2) \in {\mathbb X}\times {\mathbb X}\}$
\item[(4)]
the negativity $\{{\mathfrak N}_{n_1,n_2}\;|\;(n_1,n_2) \in {\mathbb X}\times {\mathbb X}\}$
\end{itemize}
\end{proposition}
\begin{proof}
These maps are compatible with the embeddings ${\mathfrak L}_{m_1m_2,n_1n_2}$. Indeed, for
$(m_1,m_2)|(n_1,n_2)$,
\begin{eqnarray}\label{310}
S_{n_1 n_2}[{\mathfrak L}_{m_1,n_1;m_2,n_2}(\rho _{m_1,m_2})]=S_{m_1,m_2}(\rho _{m_1,m_2})\nonumber\\
I_{n_1 n_2}[{\mathfrak L}_{m_1,n_1;m_2,n_2}(\rho _{m_1,m_2})]=I_{m_1,m_2}(\rho _{m_1,m_2})\nonumber\\
I'_{n_1|n_2}[{\mathfrak L}_{m_1,n_1;m_2,n_2}(\rho _{m_1,m_2})]=I'_{m_1|m_2}(\rho _{m_1,m_2})\nonumber\\
{\mathfrak N}_{n_1,n_2}[{\mathfrak L}_{m_1,n_1;m_2,n_2}(\rho _{m_1,m_2})]={\mathfrak N}_{m_1,m_2}
(\rho _{m_1,m_2})
\end{eqnarray}
The first and fourth of these equation is proved using the fact that the eigenvalues
of ${\mathfrak L}_{m_1,n_1;m_2,n_2}(\rho _{m_1,m_2})$ are the eigenvalues of $\rho _{m_1,m_2}$ plus
$n_1n_2-m_1m_2$ zeros (lemma \ref{500}). 
Then use of Eq.(\ref{27}) proves the second and third equation.
\end{proof}
Let $\mathbb S_2=\{S_{n_1,n_2}\;|\;n_1,n_2\in {\mathbb Z}\}$ be the set of the entropy maps, which we 
make a topological space with the divisor topology ${\mathfrak T}_{{\mathbb S}_2}$ 
(in a way analogous to that described in section \ref{topo1}).
${\mathbb S}_2$ is an ubiquitous quantity and therefore the topological space 
$(\mathbb S_2, {\mathfrak T}_{{\mathbb S}_2})$ 
describes entropy in a bipartite system and its subsystems and supersystems. 
The map ${\mathbb X}\times {\mathbb X}\;\rightarrow\;{\mathbb S}_2$ is continuous. 
Similar comments can be made for the other three quantities.

\section{Directed-complete partial orders}
\begin{definition}
A poset ${\mathbb A}$ in which every directed subset has a supremum, is a 
directed-complete partial order and is usually called dcpo.
\end{definition}
This plays an important role in domain theory\cite{D1,D2,D3} which has application in theoretical 
computer science.
The directed posets ${\mathbb X}$,  ${\mathfrak Z}$, ${\mathfrak H}$, ${\mathfrak R}$, $\Sigma$, 
are not dcpo. For example, the infinite chain $\{a,a^2,a^3,...\}$ where $a\in {\mathbb X}$, does 
not have a supremum. 
In this section we add `top elements' to the various directed posets and we make them dcpo.

\subsection{The dcpo ${\mathbb X}_1$}

${\mathbb X}_1$ is the set of supernatural (Steinitz) numbers:
\begin{eqnarray}
{\mathbb X}\subset {\mathbb X}_1=\left \{\prod p^{e_p}\;|\; p\in \Pi;\;\;\;\;
e_p\in {\mathbb Z}_0^+\cup \{\infty\}\right \}
\end{eqnarray}
If all $e_p \ne \infty$ and only a finite number of them are non-zero, 
then we get the natural numbers. 
We assume that at least one of the $e_p$ is non-zero, so that $1$ is not an element of ${\mathbb X}_1$.

Let $\pi$ be a subset of the set of prime numbers $\Pi$, and 
\begin{eqnarray}
\tau =\prod _{p\in \Pi} p^{\infty};\;\;\;\;\;
\tau (\pi)=\prod  _{p\in \pi}p^{\infty};\;\;\;\;\;\tau(\pi)|\tau
\end{eqnarray}
If $k, \ell \in {\mathbb X}_1$ we say that $k$ is a divisor of $\ell$,
when the corresponding exponents obey the relation $e_p(k)\le e_p(\ell)$, for all $p$.
This is a generalization of the usual concept of divisor.

Every element of ${\mathbb X}_1$ is a divisor of $\tau$ and therefore $\tau$ is `the top element' 
in ${\mathbb X}_1$.
Consequently ${\mathbb X}_1$ is a dcpo.
It is known that a partially order set is a dcpo if and only if each chain has a supremum\cite{D1,D2,D3}.
An example of this is the chain 
\begin{eqnarray}
p,p^2,p^3,...,p^\infty;\;\;\;\;p\in \Pi
\end{eqnarray}
which has the supernatural number $p^\infty$ as supremum.

\subsection{p-adic numbers and Pr\"ufer groups}

For later use, we summarize briefly known results about p-adic numbers, in order to establish the notation \cite{pro1,pro2,pro3}.
${\mathbb Q}_p$ is the set of p-adic numbers and ${\mathbb Z}_p$ the set of p-adic integers.
Also 
\begin{eqnarray}
\widehat {\mathbb Z}=\prod _{p\in \Pi}{\mathbb Z}_p;\;\;\;\;{\mathbb Q}/{\mathbb Z}=\prod _{p\in \Pi}{\mathbb Q}_p/{\mathbb Z}_p
\end{eqnarray}
${\mathbb Z}_p$ can be introduced as the inverse limit of the cyclic groups ${\mathbb Z}(p^n)$,
and $\widehat {\mathbb Z}$ as the inverse limit of the cyclic groups ${\mathbb Z}(n)$:
\begin{eqnarray}
\lim _{\longleftarrow}{\mathbb Z}(p^n)={\mathbb Z}_p;\;\;\;\;\;\lim _{\longleftarrow}{\mathbb Z}(n)=\widehat {\mathbb Z}\cong \prod _{p\in \Pi}{\mathbb Z}_p
\end{eqnarray}
Therefore both ${\mathbb Z}_p$ and $\widehat {\mathbb Z}$ are profinite groups.

${\mathbb Q}_p/{\mathbb Z}_p$ can be introduced as the direct limit of the cyclic groups ${\mathbb Z}(p^n)$,
and ${\mathbb Q}/{\mathbb Z}$ as the direct limit of the cyclic groups ${\mathbb Z}(n)$:
\begin{eqnarray}
\lim _{\longrightarrow}{\mathbb Z}(p^n)={\mathbb Q}_p/{\mathbb Z}_p;\;\;\;\;\;\lim _{\longrightarrow}{\mathbb Z}(n)={\mathbb Q}/{\mathbb Z}\cong \prod _{p\in \Pi}{\mathbb Q}_p/{\mathbb Z}_p.
\end{eqnarray}
The Pontryagin dual group of ${\mathbb Z}_p$ is ${\mathbb Q}_p/{\mathbb Z}_p$, and the
Pontryagin dual group of $\widehat {\mathbb Z}$ is ${\mathbb Q}/{\mathbb Z}$.

Let ${\cal C}(n)$ be the multiplicative group of the $n$-th roots of unity, which is isomorphic to the additive group ${\mathbb Z}(n)$:
\begin{eqnarray}\label{2s}
{\cal C}(n)=\{\omega _n(\alpha_n)|\alpha_n \in {\mathbb Z}(n)\}\cong {\mathbb Z}(n);\;\;\;\;n\in {\mathbb X}
\end{eqnarray}
We have the following factorization property
\begin{eqnarray}
n=\prod _{p\in \pi} p^{e_p}\;\;\rightarrow\;\;{\cal C}(n)=\prod _{p\in \pi} {\cal C}(p^{e_p})
\end{eqnarray}

We now extend this and define ${\cal C}(n)$ for all $n\in {\mathbb X}_1$.
The Pr\"ufer $p$-group ${\cal C}(p^ {\infty})$ contains 
all $p^n$-th roots of unity (for all $n \in {\mathbb Z}^+$) and it is isomorphic to ${\mathbb Q}_p/{\mathbb Z}_p$:
\begin{eqnarray}
{\cal C}(p^ {\infty})=\{\omega _n(\alpha_n)|\alpha_\ell \in {\mathbb Z}(p^n), n \in {\mathbb Z}^+\}\cong {\mathbb Q}_p/{\mathbb Z}_p
\end{eqnarray}
Its subgroups are the multiplicative cyclic groups ${\cal C}(p^n)$ (which are isomorphic to ${\mathbb Z}(p^n)$):
\begin{eqnarray}\label{2s}
{\cal C}(p)\le {\cal C}(p^2)\le ...\le {\cal C}(p^ {\infty})\cong {\mathbb Q}_p/{\mathbb Z}_p
\end{eqnarray}
More generally  the Pr\"ufer group ${\cal C}(\tau)$  is isomorphic to ${\mathbb Q}/{\mathbb Z}$:
\begin{eqnarray}
{\cal C}(\tau)=\prod _{p\in \Pi}{\cal C}(p^ {\infty})\cong \prod _{p\in \Pi}{\mathbb Q}_p/{\mathbb Z}_p \cong {\mathbb Q}/{\mathbb Z}
\end{eqnarray}
For any $n\in {\mathbb X}_1$, the ${\cal C}(n)$ is a subgroup of ${\cal C}(\tau)$.
An example is the 
\begin{eqnarray}
{\cal C}[\tau (\pi)]=\prod _{p\in \pi}{\cal C}(p^ {\infty})\cong \prod _{p\in \pi}{\mathbb Q}_p/{\mathbb Z}_p \le {\mathbb Q}/{\mathbb Z}
\end{eqnarray}

\subsection{The dcpo ${\mathfrak Z}_1$, $\Sigma _1$, ${\mathfrak H}_1$ and ${\mathfrak R}_1$}
We define the 
\begin{eqnarray}
{\mathfrak Z}\subset {\mathfrak Z}_1=\left \{{\cal C}(\alpha)\;|\; \alpha\in{\mathbb X}_1\right \}
\end{eqnarray}
All the elements in this set are subgroups of ${\cal C}(\tau)$ and  therefore ${\mathfrak Z}_1$ is a dcpo.
An example of a chain in this directed poset is the chain in Eq.(\ref{2s}) which has ${\cal C}(p^\infty)\cong {\mathbb Q}_p/{\mathbb Z}_p$ as its supremum.

We next consider a quantum system $\Sigma (\tau)$ with the position variable taking values in  
${\cal C}(\tau)\cong {\mathbb Q}/{\mathbb Z}$. We have studied such a system in ref\cite{V1}.
Its momenta take values in the Pontryagin dual group to ${\mathbb Q}/{\mathbb Z}$ which is $\widehat {\mathbb Z}$.

A subsystem of $\Sigma (\tau)$ is the system $\Sigma (p^\infty )$ with the position variable taking values in  
${\cal C}(p^\infty )\cong {\mathbb Q}_p/{\mathbb Z}_p$. We have studied such a system from both a mathematical and a physical point of view in in ref\cite{V2,V3}.
Its momenta take values in the Pontryagin dual group to ${\mathbb Q}_p/{\mathbb Z}_p$ which is ${\mathbb Z}_p$.
Another subsystem of $\Sigma(\tau)$ is the
$\Sigma[\tau (\pi)]$ which is the system with the position variable taking values in  
${\cal C}[\tau (\pi)]\cong \prod _{p\in \pi}{\mathbb Q}_p/{\mathbb Z}_p$ and
momenta taking values in the Pontryagin dual group which is $\prod_{p\in \pi}{\mathbb Z}_p$.
A summary of these systems is shown in table 1.

For any $n\in {\mathbb X}_1$, the $\Sigma(n)$ is a subsystem of $\Sigma(\tau)$.
We now consider the set of quantum systems
\begin{eqnarray}
{\Sigma}\subset {\Sigma }_1=\left \{{\Sigma}(\alpha)\;|\; \alpha\in{\mathbb X}_1\right \}
\end{eqnarray}
This is a dcpo.

In analogous way we extend the ${\mathfrak H}$ and ${\mathfrak R}$ which are not dcpo, 
into ${\mathfrak H}_1$ and ${\mathfrak R}_1$
correspondingly, which are dcpo. For example, ${\mathfrak H}_1$ will contain the space of the system $\Sigma (\tau)$
(which is described in detail in \cite{V3}).

\section{Discussion}

Using embeddings of various attributes of the system $\Sigma (m)$ into their counterparts in 
$\Sigma (n)$ (where $m|n$),
we have defined the concept `subsystem'. 
It is important that the various embeddings are compatible with each other, and we have 
shown that this is the case.
With `subsystem' as
partial order, the set of finite quantum systems $\{\Sigma (n)\}$ becomes a directed poset.

Not every quantity fits with this structure where smaller systems are embeddded into larger ones. The concept of
ubiquity aims to find the quantities that fit with this scheme.
An ubiquitous quantity has a unique value, for a state in $\Sigma (m)$ and for the corresponding state
in any of its supersystems $\Sigma (n)$ (where $m|n$).
A nonubiquitous quantity has local validity within a system $\Sigma (m)$ 
(and it needs to be recalculated if we consider its counterpart in one of its supersystems $\Sigma (n)$).
We have proved that the entropy (proposition \ref{P1}) and also the Wigner and Weyl functions are ubiquitous (proposition \ref{P2}).

We then introduced the divisor topology into the sets of various quantities. 
It is a $T_0$-topology and it encapsulates fundamental physical properties.
The open (resp. closed) set contains a quantity in some systems and in all their subsystems 
(resp. all their supersystems).
After we define topological spaces, we can define continous maps between them.

All these ideas have also been extended to bipartite systems. We have shown that various quantities 
used in the quantification of correlation and entanglement, are ubiquitous (proposition \ref{P3}).

This line of research can be extended to lattices and domains, in order to have a stronger link between
finite quantum systems and logic in the context of theoretical computation.
For this reason, in section 10, we have added `top elements' to the various posets and made them dcpo.

From a practical point of view,  
in nature there are genuine finite quantum systems (e.g., spins) and there are `pseudo-finite' quantum 
systems, which are truncations of infinite quantum systems.
An example is the Josephson qubit where the system operates in the two lowest states but it is really infinite-dimensional.
In such systems we can change the truncation point and go from $\Sigma (m)$ to $\Sigma (n)$, and then
the ideas of this paper become linked to practical problems.

\begin{table}[htbp]

\caption{Some quantum systems. The positions and momenta take values in the groups shown.}
\centering
\begin{tabular}{|c|c|c|}\hline
system & positions & momenta\\ \hline
$\Sigma (n)$ &${\mathbb Z}(n)$ & ${\mathbb Z}(n)$\\ \hline
$\Sigma (p^\infty) $&${\mathbb Q}_p/{\mathbb Z}_p$ & ${\mathbb Z}_p$\\ \hline
$\Sigma [\tau (\pi)] $&$\prod _{p\in \pi}{\mathbb Q}_p/{\mathbb Z}_p$ & $\prod _{p\in \pi}{\mathbb Z}_p$\\ \hline
$\Sigma (\tau)$ &${\mathbb Q}/{\mathbb Z}\cong \prod _{p\in \Pi}{\mathbb Q}_p/{\mathbb Z}_p$ & ${\widehat {\mathbb Z}}\cong \prod _{p\in \Pi}{\mathbb Z}_p$\\ \hline
\end{tabular}
\end{table}


\begin{thebibliography}{99}
\bibitem{FI1}
A. Vourdas, Rep. Prog. Phys. 67, 1 (2004)
\bibitem{FI2}
A. Vourdas, J. Phys. A40, R285 (2007)
\bibitem{FI3}
G. Bjork, A.B. Klimov, L.L. Sanchez-Soto, Prog. Optics 51, 469 (2008)

\bibitem{FI4}
M. Kibler, J. Phys. A42, 353001 (2009)

\bibitem{FI5}
N. Cotfas, J.P. Gazeau, J.Phys. A43, 193001(2010)

\bibitem{FI6}
T. Durt, B.G. Englert, I. Bengtsson, K. Zyczkowski, Int. J. Quantum Comp. 8, 535 (2010)

\bibitem{Horo}
R. Horodecki, R. Horodecki, R. Horodecki, R. Horodecki, Rev. Mod. Phys. 81, 865 (2009)

\bibitem{counterexamples}
L.A. Steen, J.A. Seebach, `Counterexamples in topology' (Dover, New York, 1995)

\bibitem{top1}
J.L. Kelly, `General Topology', (Springer, Berlin, 1955) 

\bibitem{top2}
L. Nachbin, `Topology and Order' (Van Nostrand, 1965)

\bibitem{D1}
D.S. Scott, Lecture Notes in Mathematics 274, 97 (1972)

\bibitem{D2}
S. Abramsky, A. Jung, in S. Abramsky, D.M. Gabbay, T.S.E. Malbaum (Eds), Handbook of logic in Computer Science, 
(Oxford Univ. Press, Oxford, 1994)
 
\bibitem{D3}
G. Gierz, K.H. Hofmann, K. Keimel, J.D. Lawson, M.W. Mislove, D.S. Scott, 
`Continuous lattices and domains' (Cambridge Univ. Press, Cambridge, 2003) 


\bibitem{AR}
H. Bateman, `Higher transcedental functions' Vol.3 (McGraw-Hill,  New York, 1955)

\bibitem{VB}
A. Vourdas, C. Banderier, J. Phys. A43, 042001 (2010) (and corrigendum J. Phys. A44, 149501 (2011)

\bibitem{SV}
M. Shalaby, A. Vourdas, J. Phys. A45, 052001 (2012)

\bibitem{Willard}
S. Willard, `General topology' (Dover, New York, 1970)

\bibitem{BOU}
N. Bourbaki, `General Topology', Part 1 (Hermann, Paris, 1966)

\bibitem{Po}
L.S. Pontryagin, `Topological groups' (Gordon and Breach, New York, 1966)

\bibitem{AC}
N.J. Cerf, C. Adami, Physica D120, 62(1998)
\bibitem{VW}
G. Vidal, R.F. Werner, Phys. Rev. A65, 032314 (2002)


\bibitem{pro1}
L. Ribes, P. Zalesskii, `Profinite groups', Springer, Berlin, 2000

\bibitem{pro2}
J. Wilson, `Profinite groups', (Clarendon, Oxford, 1998)

\bibitem{pro3}
D. Ramakrishnan, R.J. Valenza, `Fourier analysis on number fields', Springer, Berlin, 1999

\bibitem{V1}
A. Vourdas, J. Phys. A41, 455303 (2008)

\bibitem{V2}
A. Vourdas, J. Fourier Anal. Appl. 16 , 748 (2010) 

\bibitem{V3}
A. Vourdas, J. Math.Phys. 52, 062103 (2011)
\end{thebibliography}
\end{document}